\documentclass[11pt]{article}
\usepackage{latexsym}
\usepackage[pdftex]{graphicx}
\usepackage{amsfonts}
\usepackage{amsmath, amsthm, amssymb}
\usepackage{fullpage}
\usepackage{setspace}
\usepackage{longtable}
\usepackage{natbib}
\bibpunct{(}{)}{;}{a}{}{,}
\usepackage{dcolumn}
\newcolumntype{.}{D{.}{.}{-1}}
\newcolumntype{d}[1]{D{.}{.}{#1}}
\usepackage[top=.8in, left=.9in, right=.9in,
bottom=.8in]{geometry}
\usepackage{bm}
\usepackage[compact]{titlesec}
\usepackage{booktabs}

\usepackage{color}
\RequirePackage[colorlinks,citecolor=blue,urlcolor=blue]{hyperref}
\usepackage{amsthm}

\newtheorem{theorem}{Theorem}
\newtheorem{lemma}{Lemma}

\usepackage{caption}
\usepackage{subcaption}

\def\hat{\widehat}
\def\tilde{\widetilde}
\def\P{{\rm pr}}
\def\bX{{X}}
\def\bfW{{W}}
\def\bY{{Y}}
\def\bI{I}
\def\bfZ{{Z}}
\def\bfz{{z}}
\def\var{\mbox{var}}
\def\calA{{\mathcal D}}

\def\mI{{\cal I}}
\def\bfbeta{{\beta}}

\expandafter\def\expandafter\normalsize\expandafter{\normalsize\setlength\abovedisplayskip{0pt}}
\expandafter\def\expandafter\normalsize\expandafter{\normalsize\setlength\belowdisplayskip{0pt}}
\expandafter\def\expandafter\normalsize\expandafter{\normalsize\setlength\abovedisplayshortskip{0pt}}
\expandafter\def\expandafter\normalsize\expandafter{\normalsize\setlength\abovedisplayshortskip{0pt}}

\begin{document}
\pagestyle{plain}

\title{Inference on Average Treatment Effect under Minimization and Other Covariate-Adaptive Randomization Methods}
\author{Ting Ye\thanks{Department of Statistics, University of Pennsylvania} \and Yanyao Yi\thanks{Global Statistical Sciences, Eli Lilly and Company}  \and Jun Shao\thanks{School of Statistics, East China Normal University}}

\maketitle
\thispagestyle{empty}

\abstract{
	Covariate-adaptive randomization schemes such as the minimization and stratified permuted blocks are often applied in clinical trials to balance treatment assignments across  prognostic factors. The existing theoretical developments on inference after covariate-adaptive randomization are mostly limited to situations where a correct model between the response and covariates can be specified or the randomization method has well-understood properties.  Based on stratification with covariate levels utilized in randomization and a further adjusting for covariates not used in randomization, in this article we propose several estimators for model free  inference on average treatment effect  defined as the difference between response means under two treatments. We establish asymptotic normality of the proposed estimators under all popular covariate-adaptive randomization schemes including the minimization
whose theoretical property is unclear, and we show that the asymptotic distributions are invariant with respect to covariate-adaptive randomization methods. Consistent variance estimators are constructed for asymptotic inference.  Asymptotic relative efficiencies and finite sample properties of estimators   are also studied. We recommend  using one of our proposed estimators for valid and model free inference after covariate-adaptive randomization. 
}

\newcommand{\n}{\noindent}

\clearpage

\section{Introduction}

Consider a clinical trial to compare $k$ treatments  with given treatment assignment proportions $\pi_1,\ldots,\pi_{k}$, where $k\geq 2$ is a fixed positive integer, $ \sum_{t=1}^k\pi_t=1 $, and $\pi_t$ can be any known number strictly between 0 and 1. In many trials patients are not all available for simultaneous assignment of treatments but rather 
arrive sequentially and must be treated immediately. Thus, 
simple randomization, which assigns patients to treatments completely at random,  may yield  sample sizes not following the assignment proportions across  prognostic factors or covariates, e.g., institution, disease stage, prior treatment, gender, and age, which are thought to have significant influence on the responses of interest.  
A remedy is to apply covariate-adaptive randomization, i.e., treatment assignment of the $i$th  patient depends on the observed covariate value of this patient and the assignments and covariate values of all $i-1$ previously assigned  patients. 
In this article, we focus on  enforcing  assignment allocation across 
levels of a covariate vector $\bfZ$  whose  components are discrete or discretized continuous covariates. There are model-based approaches of balancing discrete or continuous covariates for estimation efficiency  \citep{Atkinson:1982aa,Atkinson:1999aa,Atkinson:2002aa,Rosenberger:2008aa,Senn:2010aa,Baldi-Antognini:2011aa}, which  are not  further considered in this article.
The oldest  method of balancing covariates is the minimization  \citep{Taves:1974aa} intended to 
balance treatment assignments across marginal levels of $Z$: 
it  assigns the $i$th  patient by minimizing a weighted sum of squared or absolute differences between the numbers of patients, up to the $i$th, 
assigned to treatments over marginal levels of $\bfZ$. 
\cite{Pocock:1975aa} extended 
Taves' procedure to 
achieving minimization with a given probability, 
which is still referred to as the minimization method.  
Other popular covariate-adaptive randomization methods include 
the stratified permuted block randomization  \citep{Zelen:1974aa}, the 
stratified biased coin \citep{Shao:2010aa,Kuznetsova:2017aa}, and
the stratified urn design \citep{Wei:1977aa,Zhao:2016aa}. 
See  \cite{Schulz:2018aa} and \cite{Rosenberger:2008aa} for nice summaries.
As pointed out in \cite{Taves:2010aa},  from 1989 to 2008, over 500 clinical trials  implemented the minimization method to balance important covariates, despite some criticisms by  \cite{Smith:1984aa} and \cite{Senn:2010aa}.  According to a recent review of nearly 300 clinical trials published in  2009 and 2014  \citep{Ciolino:2019aa}, 237 of them used covariate-adaptive randomization.


Although data are collected under covariate-adaptive randomization, conventional inference procedures constructed based on simple randomization are often applied in practice. 
This has raised concerns because statistical inference on  treatment effects should be made using procedures valid under the particular randomization scheme used in data collection. Applying conventional inference after covariate-adaptive randomization may lead to invalid results, especially for the minimization because its theoretical property remains largely unclear. In its 2015  guideline, \cite{ema:2015aa} raised concerns and specifically pointed out that ``possible implications of dynamic allocation methods [minimization] on the analysis e.g. with regard to bias and Type I error control should be carefully considered, \ldots conventional statistical methods do not always control the Type I error''.

Starting with \cite{Shao:2010aa}, there has been significant progress in understanding the theoretical properties of statistical tests under 
covariate-adaptive randomization, e.g.,  
\cite{Hu:2012aa},
\cite{Shao:2013aa}, \cite{Ma:2015aa}, \cite{Bugni:2017aa,Bugni:2019aa}, \cite{Ye:2018aa},   and \cite{Ye:2020aa}.  Another important stream of statistical inference methods is based on permutation tests or re-randomization inference, e.g., \cite{Simon:2011aa}, \cite{Kaiser:2012aa}, and \cite{Bugni:2017aa}.
However, except for \cite{Bugni:2019aa},  all theoretical results are established  under the assumption that  either a correct model between the responses of interest  and covariates  is available or  the covariate-adaptive randomization procedure has a well-understood property, described as  type 1 or type 2 later in \S2 of the current paper. 
It should be noted  that  model misspecification often occurs, especially when there are many covariates, and
the minimization method is neither type 1 nor type 2. 
The minimization  is applied in practice very often 
\citep{Pocock:1975aa}, mainly because  
it  aims to minimize the imbalance across marginal levels of $\bfZ$, not every joint level of $\bfZ$, which is sufficient in many applications. 
Enforcing treatment balance in every joint level of $\bfZ$ may cause sparsity of data when  the dimension of $\bfZ$ is not small. 

To fill the gap, in this paper we propose  asymptotically valid inference on the average treatment effect defined as the difference between population response means of every treatment pair, under covariate-adaptive randomization including  minimization.
Our main idea is to apply stratification according to the  levels of discrete $\bfZ$ and 
to  adjust for covariates not used in treatment randomization through generalized regression for improving efficiency.
Our estimator without adjusting  for covariates, 
which is not the most efficient one, coincides with the estimator  derived under a different approach in the recent publication \citep{Bugni:2019aa}.  
Asymptotic normality of the proposed treatment effect estimators is established with explicit limiting variance formulas that can be used for inference as well as comparing relative efficiencies.   Our results are not only model-free,  
i.e.,  only the existence of second-order moments of the responses and covariates are required,  but also invariant with respect to covariate-adaptive  randomization schemes, i.e., the same inference procedure can be applied under any 
covariate-adaptive randomization. 
We also study and compare inference procedures by simulations
and illustrate our method in a real data example.
R codes for the methods proposed in this paper can be found in the R package RobinCar  posed at \texttt{https://github.com/tye27/RobinCar}.


\section{Preliminaries}


Let $\bI$  be the treatment indicator vector, which equals $ e_t $ if the patient is assigned to treatment $t$, $t=1,\ldots,k$, where $ e_t $ denotes a vector with only the $ t $th component being 1 while the rest being 0. Let 
$Y^{(t)}$  be the potential response under treatment $t$,  $\bfW$ be a  vector of all observed covariates, and 
$\bfZ$ be a discrete function of $\bfW$ utilized in covariate-adaptive randomization.  
For patient $i$, 
let $\bI_i$, $ \bfW_i $, and $\bY^{(t)}_i, t=1,\ldots, k$,  be realizations of $\bI$, $\bfW$, and $\bY^{(t)}$, $t=1,\ldots, k$, respectively, where $i=1,\ldots,n$, and  $n$ is the total number of patients in all treatment arms. 
For every patient $i$, $\bI_i$ is generated after $\bfZ_i$ is observed,   and only  the potential response from the treatment indicated by $\bI_i$ is observed, i.e., we observe  $Y_i = Y_i^{(t)}$ if and only if $I_i = e_t$.

After all treatment assignments are made and responses are collected, we would like to make inference based on the observed data $\{ \bfW_i, \bI_i, Y_i,  i=1,\ldots,n \}$. 
In this paper, we consider inference on the average treatment effect between treatments $t$ and $s$, \vspace{-1mm}
\begin{equation}
\theta =E(Y^{(t)} -Y^{(s)} ), \label{theta} 
\end{equation}
where $E$ denotes the population expectation and $s \neq t$ are given integers between $ 1 $ and $k$. Note that $\theta$ in (\ref{theta}) depends on $t$ and $s$, but a subscript is omitted for simplicity, as we mainly consider the average treatment effect for two fixed treatment arms. 

For our inference procedure studied in \S 3, we describe some minimal conditions. 
The first one is about the population for potential responses and covariates. 
\begin{description}
	\item
	(C1)  
	$(\bY_i^{(1)},\ldots, \bY_i^{(k)}, \bfW_i)$, $i=1,\ldots, n$, are independent and identically distributed as\\ $(\bY^{(1)},\ldots, \bY^{(k)}, \bfW)$ and $ Y^{(t)} $ has finite second-order moments, $ t=1,\ldots, k $.
\end{description}
This condition  is model-free: there is no assumption on the relationship between $\bfW$ and the potential response $Y^{(t)}$ that 
may be  continuous or discrete. 


Under simple randomization, $\bI_i$'s are independent of $ (\bY_i^{(1)},\ldots, \bY_i^{(k)}, \bfW_i)$'s and are independent and identically distributed with $\P(\bI_i =e_t)= \pi_t$.  
To enforce assignment proportions at each joint level of $\bfZ$ treated as stratum, 
three popular covariate-adaptive randomization schemes are the stratified permuted block randomization method \citep{Zelen:1974aa}, the stratified 
biased coin method \citep{Shao:2010aa,Kuznetsova:2017aa}, and 
the stratified urn design \citep{Wei:1977aa,Zhao:2016aa}.

The minimization \citep{Taves:1974aa, Pocock:1975aa,Han:2009aa} is the same as the stratified biased coin method if $\bfZ$ is one-dimensional, but is very different from the above three stratification methods with  a multivariate $\bfZ$.  
It aims to enforce assignment ratio across marginal levels of $\bfZ$, not every stratum defined by the joint level of $\bfZ$. 
Assignments are made  by minimizing a weighted sum of squared or absolute differences between the numbers of  patients assigned to treatment arms across marginal levels of $\bfZ$. 
Because only marginal levels of $\bfZ$ are considered in minimization, 
this method is also called the marginal method in   \cite{Ma:2015aa}  and \cite{Ye:2020aa}. 

We assume the following minimal conditions for covariate-adaptive randomization. 
\begin{description}
	\item
	(C2)
	$(\bI_i, i=1,\ldots,n)$ and $(\bY_i^{(1)},\ldots, \bY_i^{(k)}, \bfW_i , i=1,\ldots,n)$ are conditionally independent given $\bfZ_1,\ldots,\bfZ_n$. 
	\item
	(C3) 
	$\bfZ$ is discrete with finitely many levels given in a set ${\cal Z}$. For each $t=1,\ldots,k$, $\P(\bI_i=e_t \mid\bfZ_1, \ldots, \bfZ_n)=\pi_t$ for $i=1,\ldots, n$, and  $\{n(\bfz)\}^{-1}D_t(\bfz)$ converges to 0 in probability as $n\rightarrow\infty$ for every $\bfz \in {\cal Z}$, where $n(\bfz)$  is the  number of patients with $\bfZ_i = \bfz$, and $D_t(\bfz)= n_t(z) - \pi_t n(\bfz )$ with $n_t(z)$ being the number of patients with $\bfZ_i = \bfz$ assigned to treatment $t$.  
\end{description}
Condition (C2) is reasonable because (i) given  $\bfZ_i$'s, $\bfW_i$'s  contain covariates not used in randomization, and (ii) treatment assignments do not affect the  potential responses, although they do affect the observed responses $Y_i$'s.  
Condition (C3) holds for most covariate-adaptive randomization schemes
\citep{Baldi-Antognini:2015aa}, certainly for all schemes considered in this paper, 
the minimization and three stratified designs, the permuted block,
biased coin, and urn design.  


We classify 
all covariate-adaptive randomization methods into the following three types  in terms of  $D_t (\bfz )$  defined in 
(C3). 

\begin{description}
	\item []Type 1.  For every $t$ and $\bfz$, $\{n(\bfz )\}^{-1/2} D_t (\bfz ) \to 0$ in probability as $n\rightarrow\infty$. 
	\item []Type 2.  For every $t$, $D_t(\bfz )$, $\bfz \in {\cal Z}$, are independent and, for every $t$ and $\bfz $, $\{n(\bfz )\}^{-1/2} D_t(\bfz )\xrightarrow{d}  N(0, v_t)$ with a known $v_t>0$, where $\xrightarrow{d} $ denotes convergence in distribution as $n \to \infty$. 
	\item []Type 3. Methods not  in type 1 or 2. 
\end{description}

The three types are  defined based on their degree in enforcing the balancedness according to the given assignment proportions within every joint level of $\bfZ$.  Type 1 is the strongest, since $D_t(\bfz)$ measures the imbalance of assignments within stratum $\bfz$. The property 
$\{n(\bfz )\}^{-1/2} D_t(\bfz )\to 0$ in probability is stronger than $\{n(\bfz)\}^{-1}D_t(\bfz)\to 0$ in probability in (C3). 
Type 2 is  weaker than type 1 in
enforcing the balancedness, as it requires $\{n(\bfz )\}^{-1/2} D_t(\bfz )$ converging in distribution, not in probability to 0, although it is still stronger than  
$\{n(\bfz)\}^{-1}D_t(\bfz)\to 0$ in probability. 

Representatives of type 1 methods are  
stratified permuted block randomization and stratified biased coin methods. 
Specifically, under stratified permuted block randomization,  $D_t(\bfz)$  is bounded by the maximum block size. For the stratified  biased coin method, it follows from a result in \cite{Efron:1971aa} that $D_t(\bfz)$ is bounded in probability.
The  stratified urn design is type 2 with $v_t= 1/12$  when $k=2$ and  $\pi_1=\pi_2 = 1/2$ \citep{Wei:1978ab}.  Simple randomization treated as a special case of covariate-adaptive randomization is also type 2.  
Finally, the minimization is type 3, since it is neither type 1 nor type 2 \citep{Ye:2020aa}. Specifically, under minimization,  $D_t(\bfz )$ and $D_t(\bfz')$ with $\bfz \neq \bfz'$ are not independent, and their relationship is complicated,  
because assignments are made according to marginal levels of $\bfZ$. 

For type 1 methods, some theoretical results in statistical testing have been established; see, for example,   
\cite{Shao:2010aa},  \cite{Shao:2013aa},  \cite{Ye:2018aa},   \cite{Bugni:2017aa, Bugni:2019aa}, 
and \cite{Ye:2020aa}. \cite{Bugni:2017aa, Bugni:2019aa} and \cite{Ye:2020aa}
also  considered type 2 methods. 
In the next section, we propose inference procedures and establish their asymptotic validity under general covariate-adaptive randomization including minimization.

\section{Inference on Average Treatment Effect}


To make asymptotically valid inference on $\theta$ defined in (\ref{theta}), 
the key is to construct an estimator of $\theta$ and derive its asymptotic distribution. 
Under simple randomization, the simplest estimator is the response mean difference  $\bar{Y}_t-\bar{Y}_s$, where $\bar{Y}_t$ is the sample mean of responses under treatment $t$. Although $\bar{Y}_t-\bar{Y}_s$ is asymptotically normal under type 1 or 2 covariate-adaptive randomization, it is generally not efficient as covariate information is not utilized in estimation. More seriously, 
the asymptotic distribution of $\bar{Y}_t-\bar{Y}_s$ is not known under type 3 
covariate-adaptive randomization such as the minimization. 
\cite{Bugni:2017aa}  derived a different estimator of $\theta$, called the strata fixed effect estimator in their \S4.2, but its asymptotic normality is established only for type 1 or 2 covariate-adaptive randomization and, thus, it cannot be used  under type 3 randomization such as the minimization. 

Let $\bar{Y}_t (\bfz )$ be the sample mean of $Y_i$'s with $\bfZ_i = \bfz$ under treatment $t$. 
The following stratified  response mean differences with strata being all joint levels of $\bfZ$ is proposed in expression (8) of \cite{Bugni:2019aa}, 
\begin{eqnarray}
\hat{\theta} =\sum_{\bfz \in {\cal Z}} 
\frac{n (\bfz)}{n}\{ \bar{Y}_t(\bfz)-\bar{Y}_s(\bfz)\}, \label{eq: theta hat}
\end{eqnarray}
although  \cite{Bugni:2019aa}
provided $\hat\theta$ in a different form derived under a fully saturated linear regression. 
If the weight $n(\bfz )/n$ in (\ref{eq: theta hat}) is replaced by the population weight  $\P( \bfZ = \bfz )$, then $\hat\theta$ is exactly the stratified estimator in survey sampling. 
We use $n(\bfz )/n$ in (\ref{eq: theta hat})
as $\P( \bfZ = \bfz )$ is unknown.

Although $\hat\theta$ in (\ref{eq: theta hat}) utilizes information from $\bfZ$ by stratification and is asymptotically more efficient than 
the simple $\bar{Y}_t-\bar{Y}_s$ or the strata fixed effect estimator in \cite{Bugni:2017aa}, 
it does not make use of covariate information in $\bfW$ but not in $\bfZ$.
Note that $\bfW$ may contain  components that are not in $\bfZ$ but are related with the potential responses $ Y^{(t)}, t=1,\ldots, k $, or some components of $\bfZ$ are discretized components of $\bfW$ and the remaining information after discretization is still predictive of $ Y^{(t)}, t=1,\ldots, k $.  


Let $\bX$ be a function of $\bfW$ that we want to further adjust for. 
We now consider improving $\hat\theta$ in (\ref{eq: theta hat}) by utilizing $\bX$. 
To maintain model free estimation, we do not impose any  model between $ Y^{(t)}$ and $\bX$, but adjust for covariate $\bX$ within each $\bfZ = \bfz$ by applying the  generalized regression approach in survey sampling, first discussed in \cite{CASSEL:1976aa}
and studied extensively in the literature, for example,   \cite{Sarndal:2003aa}, 
\cite{Lin:2013aa}, \cite{Shao:2014aa}, and \cite{Ta:2020aa}.   Since this approach  is model-assisted but not model-based, i.e., a model is used to derive efficient estimators that are still asymptotically valid  even if the model is incorrect,  it suits our purpose 
of utilizing covariates without modeling. 

Let $\bX_i$ be the value of covariate $ \bX $ for patient $i$,
$\bar{\bX}_t(\bfz )$ be the sample mean of $\bX_i$'s  with $Z_i=z$ under treatment $ t $, and \vspace{-2mm}
\[ \hat{\beta}_t(\bfz ) = \left[\sum_{ i: I_i=e_t, Z_i=z } \{\bX_i - \bar{\bX}_t(\bfz )\}\{\bX_i - \bar{\bX}_t(\bfz )\}^T
\right]^{-1}  \sum_{ i: I_i=e_t, Z_i=z} \{\bX_i - \bar{\bX}_t(\bfz )\} Y_i,  \vspace{-1mm} \]
where $a^T$ is the  transpose of  vector $a$.  
Within treatment $t$ and $\bfZ = \bfz$, $ \hat{\bfbeta}_t(\bfz ) $ is 
the least squares estimator of the coefficient vector in front of $\bX$ under a linear  model between $Y^{(t)}$ and $\bX$, but the model is not required to be  correct. 
Then, our first proposed estimator of $\theta$ after adjusting for covariates is 
\begin{equation}
\hat\theta_{A} =\sum_{\bfz \in {\cal Z}} 
\frac{n (\bfz)}{n}\{ \bar{Y}_{t,A}(\bfz)-
\bar{Y}_{s,A}(\bfz)\}, \qquad 
\bar{Y}_{t,A} (\bfz ) = \bar{Y}_t(\bfz ) -  \{\bar{\bX}_t(\bfz ) - \bar{\bX}(\bfz ) \}^{T} \hat{\bfbeta}_t(\bfz ), \label{adj1}
\end{equation}
where 
$\bar{\bX} (\bfz )$ is the sample mean of $\bX_i$'s of all patients with $\bfZ_i = \bfz$. 


Within $\bfZ_i = \bfz$, 
if we assume that the linear models under all treatments have the same  coefficient vector for $X$,  such as  a homogeneous ANCOVA model with  covariate vector $\bX$, then we  can replace $\hat\beta_t(\bfz)$ by  \vspace{-1mm}
\[ \hat{\bfbeta} (\bfz ) = \left[\sum_{t=1}^{k} \ \sum_{ i: I_i=e_t, Z_i=z } \{\bX_i - \bar{\bX}_t(\bfz )\}\{\bX_i - \bar{\bX}_t(\bfz )\}^{T}
\right]^{-1} \sum_{t=1}^{k} \ \sum_{ i: I_i=e_t, Z_i=z } \{\bX_i - \bar{\bX}_t(\bfz )\} Y_i.   \vspace{-1mm}\]
Again, the model is not required to be correct in order to use $\hat{\bfbeta} (\bfz ) $.  This leads to an alternative estimator of $\theta$ after adjusting for covariates, 
\begin{equation}
\hat\theta_{B}   =\sum_{\bfz \in {\cal Z}} 
\frac{n (\bfz)}{n}\{ \bar{Y}_{t,B}(\bfz)-\bar{Y}_{s,B}(\bfz)\}, \qquad 
\bar{Y}_{t,B} (\bfz ) = \bar{Y}_t(\bfz ) -  \{\bar{\bX}_t(\bfz ) - \bar{\bX}(\bfz ) \}^{T} \hat{\bfbeta} (\bfz ).  \label{adj2}
\end{equation}

When $k >2$, both $\bar{X}(z)$ and $\hat\beta (z)$ involve data from patients in treatment arms other than $t$ and $s$.

The following theorem proved in the Supplementary Material derives the asymptotic distributions of $\hat\theta$ in (\ref{eq: theta hat}), 
$\hat\theta_{A}$ in (\ref{adj1}), and $\hat\theta_{B}$ in (\ref{adj2}), under covariate-adaptive randomization including minimization.  

\begin{theorem} \label{theo1}
	Assume (C1)-(C3), the existence of second-order moments of $ X$, $XY^{(t)}, t=1,\ldots, k $, and that ${\rm var} (\bX \mid \bfZ= \bfz )$ is positive definite for every $\bfz \in {\cal Z}$. As $n \to \infty$, 
	\begin{align*}
	&		\sqrt{n} (\hat{\theta}-\theta) \ \xrightarrow{d}  \ N(0, {\sigma}_{U}^2+\sigma_V^2), \\ 
	&\sqrt{n} (\hat{\theta}_A -\theta)\xrightarrow{d}  N(0, {\sigma}_{A}^2+\sigma_V^2), \\
	& \sqrt{n} (\hat{\theta}_B -\theta)\xrightarrow{d}  N(0, {\sigma}_{B}^2+\sigma_V^2), 
	\end{align*}\vspace{-1mm}
	where 	 \vspace{-1mm}
	\begin{align*}
	\sigma_U^2= & \, E\{ {\rm var} (Y^{(t)}\mid\bfZ)/\pi_t+{\rm var} (Y^{(s)}\mid\bfZ)/\pi_s\}, \\
	{\sigma}_{A}^2 = & \, E[{\rm var}\{Y^{(t)} - \bX^{T} \bfbeta_t(\bfZ) \mid\bfZ\}/\pi_t + {\rm var}\{ Y^{(s)}- \bX^{T}\bfbeta_s(\bfZ) \mid\bfZ\}/ \pi_s] \\
	& + E[\{ \bfbeta_t(\bfZ)-\bfbeta_s(\bfZ)\}^{T} {\rm var}(\bX\mid\bfZ)
	\{ \bfbeta_t(\bfZ)-\bfbeta_s(\bfZ)\}],\\
	{\sigma}_{B}^2= & \, E[ {\rm var}\{Y^{(t)} - \bX^{T} \bfbeta (\bfZ) \mid\bfZ\}/\pi_t + {\rm var}\{Y^{(s)}- \bX^{T} \bfbeta (\bfZ) \mid\bfZ\}/ \pi_s], \\
	\sigma_V^2= & \, {\rm var}\{ E(Y^{(t)}-Y^{(s)}\mid\bfZ)\},  \vspace{-1mm}
	\end{align*}
	$t, s =1,\ldots,k$,	$	 \bfbeta_t(\bfz)=  \{{\rm var}(\bX\mid\bfZ=\bfz)\}^{-1}
	{\rm cov}(\bX, Y^{(t)}\mid\bfZ=\bfz) $,  and $\bfbeta(\bfz)= \sum_{t=1}^k \pi_t\bfbeta_t(\bfz)$. 
\end{theorem}

Theorem 1 is  model free and is applicable to  any covariate-adaptive randomization method satisfying (C2)-(C3), most noticeably the  minimization for which very little is known about its theoretical property, as the minimization is neither type 1 nor type 2 described in \S2. This provides a solid foundation for valid and model free inference after minimization. 

The asymptotic result in Theorem 1 is invariant with respect to randomization methods, i.e., $\sigma_U^2$, $\sigma_A^2$, $\sigma_B^2$, and $\sigma_V^2$ do not depend on the randomization scheme. In other words, each estimator of $\theta$ in (\ref{eq: theta hat})-(\ref{adj2}) has the same asymptotic distribution and efficiency regardless of which randomization scheme is used for treatment assignments, including simple randomization. 
This is intrinsically different from many existing results that  are  dependent on randomization methods
\citep{Shao:2013aa,Ma:2015aa,Bugni:2017aa}.
The only result  invariant  with respect to randomization methods can be found in the literature is in \cite{Bugni:2019aa} for $\hat\theta$ in (\ref{eq: theta hat}), although it does not explicitly state this invariance property. 

Due to the use of covariate-adaptive randomization, the sample mean $\bar{Y}_t$ is not an average of independent random variables and, thus, the  asymptotic distributions of estimators in (\ref{eq: theta hat})-(\ref{adj2}) cannot be obtained by directly applying the central limit theorem for sum of independent random variables. We overcome this difficulty by decomposing $\hat{\theta}-\theta$  as the sum of the following two uncorrelated terms,\vspace{-1mm}
\begin{align*}
U= &\sum_{\bfz \in {\cal Z}} \frac{n(\bfz)}{n}\left[ \{ \bar{Y}_t(\bfz)-\bar{Y}_s(\bfz)\} -\{ E(Y^{(t)}\mid\bfZ=\bfz) -E(Y^{(s)}\mid\bfZ=\bfz) \}\right],\\
V= &\sum_{\bfz \in {\cal Z}} \frac{n(\bfz)}{n} \{ E(Y^{(t)}\mid\bfZ=\bfz) -E(Y^{(s)}\mid\bfZ=\bfz) \} -\theta.
\end{align*}
Conditioned on $(I_1,\ldots, I_n, Z_1,\ldots,Z_n)$, $U$ is an average of independent terms so its limiting distribution can be derived by applying the central limit theorem, which consequently provides the unconditional asymptotic distribution of $U$. For $V$, the only random part is $n(\bfz )$ whose limiting distribution can be easily derived. 
For $\hat\theta_A $ or $\hat\theta_B $, a similar decomposition can be obtained with the same $V$ and a different $U$ incorporating the covariate adjustment term.  Details can be found in the Supplementary Material.

This decomposition  is not only   the key to establishing the 
asymptotic result, but also identifies two sources of variation.
The variation of potential responses $Y^{(t)}$ and $Y^{(s)}$ explained by $\bfZ$ is represented by $\sigma_U^2$. 
The variation from  treatment effect heterogeneity is measured by 
$\sigma_V^2$.
Note that  we allow arbitrary treatment effect heterogeneity, i.e.,  different 
subgroups according to levels of $\bfZ$ may benefit differently from the treatment. 
If there is no treatment effect heterogeneity, then  	$\sigma_V^2=0$.

The asymptotic relative efficiencies among $\hat\theta$, $\hat\theta_A$ and $\hat\theta_B$ are summarized in the following result. 

\begin{theorem} \label{theo2} Under the same assumptions  in Theorem 1, for $\sigma^2_U$, $\sigma^2_A$, $\sigma^2_B$, $\bfbeta_t (\bfz )$ and $\bfbeta (\bfz )$ defined in Theorem 1, we have 
	\begin{align*}
	\sigma_U^2 - \sigma_A^2   = & \,	 E\left[  \{\pi_s \bfbeta_t(\bfZ)+\pi_t\bfbeta_s(\bfZ)\}^{T} {\rm var}( \bX \mid \bfZ )\{\pi_s \bfbeta_t(\bfZ)+\pi_t\bfbeta_s(\bfZ)\}\right]\{\pi_t\pi_s(\pi_t+\pi_s)\}^{-1}  \\
	& + E\left[ \left\{\beta_t(Z)- \beta_s(Z)\right\}^T{\rm var}(X\mid Z)  \left\{\beta_t(Z)- \beta_s(Z)\right\}\right] \{ (\pi_t+\pi_s)^{-1}-1\}, \\
	\sigma_B^2 - \sigma_A^2 = & \, E\left[ \{\beta_t(Z)- \beta(Z)\}^T {\rm var}( \bX \mid \bfZ ) \{\beta_t(Z)- \beta(Z)\}\right]\pi_t^{-1}\\
	&+E\left[ \{\beta_s(Z)- \beta(Z)\}^T {\rm var}( \bX \mid \bfZ ) \{\beta_s(Z)- \beta(Z)\}\right]	\pi_s^{-1} \\
	&- E\left[ \left\{\beta_t(Z)- \beta_s(Z)\right\}^T{\rm var}(X\mid Z)  \left\{\beta_t(Z)- \beta_s(Z)\right\}\right].    
	\end{align*}
	Consequently,   ${\sigma}_{A}^2\leq \sigma_U^2$, where the equality holds if and only  if for every $\bfz \in {\cal Z}$,\vspace{-1mm}
	\begin{equation}
	\pi_s\beta_t(z)+\pi_t\beta_s(z)=0  \quad  \mbox{and} \quad   \{\beta_t(z)- \beta_s(z)\}(1-\pi_t-\pi_s)=0;  \label{equiv1}
	\end{equation}
	and ${\sigma}_{A}^2\leq {\sigma}_{B}^2$,  where the equality holds if and only if for every $\bfz \in {\cal Z}$, \vspace{-1mm}
	\begin{equation}
	\beta(z)=\{\pi_s\beta_t(z)+\pi_t\beta_s(z) \}/ (\pi_s+\pi_t) \quad  \mbox{and} \quad  \{\beta_t(z)- \beta_s(z)\}(1-\pi_t-\pi_s)=0. \label{equiv2}
	\end{equation}
\end{theorem}

Theorem \ref{theo2} indicates that $\hat\theta_A$ is always asymptotically more efficient than $\hat\theta$ unless (\ref{equiv1}) holds, in which case $\hat\theta$ and $\hat\theta_A$ have the same asymptotic efficiency. 
This theoretically corroborates the perception that
covariate adjustment with a full set of treatment-covariate interactions can not hurt efficiency. 
When there are more than two treatments, $1-\pi_t - \pi_s >0$ and, consequently,  (\ref{equiv1}) holds  only when 
$\beta_t(z)= \beta_s(z)=0 $ for every $z$, i.e., 
$\bX$ is uncorrelated with the potential responses $Y^{(t)}$ and $Y^{(s)}$ after conditioning on $\bfZ$ so that adjusting for $\bX$ is unnecessary. 
When there are only two treatments,
(\ref{equiv1}) also holds if  $\pi_t=\pi_s= 1/2$ and 
$\beta_t(z)=  - \beta_s(z) $ for every $z$. An example is
given in \S 4.

Note that $\hat{\bfbeta} (\bfz )$ used in $\hat\theta_B$ ignores the fact that ${\rm cov}(\bX, Y^{(t)}\mid\bfZ=\bfz)$ may  depend on treatment $t$. That is why
$\hat\theta_B$ is  asymptotically not as efficient as $\hat\theta_A$ in general, and $\sigma^2_B = \sigma^2_A$ when these covariances are the same for every $t$ and every  $\bfz$, i.e., $ \beta_1(z)=\cdots= \beta_k(z) $. An exceptional case is that  $ \sigma_A^2=\sigma_B^2 $ when there are only two treatments and $ \pi_t=\pi_s=1/2 $. In fact, 
$\hat\theta_B$  may be asymptotically less efficient than $\hat\theta$, i.e., covariate adjustment with only the main effects may hurt efficiency, a  perspective in \cite{Freedman:2008aa} and \cite{Lin:2013aa}. 
For example, there are scenarios in which  (\ref{equiv1}) holds but  (\ref{equiv2}) does not.
Simulation examples are given in \S4, where comparisons of $\hat\theta$, $\hat\theta_A$, and $\hat\theta_B$ are made.  



To make model free inference on the average treatment effect $\theta$ defined in (\ref{theta}), we only need to apply Theorem 1 and construct consistent estimators of limiting variances  $\sigma^2_U$, $\sigma^2_A$, $\sigma^2_B$, and $\sigma^2_V$. 
Let  $S_t^2(\bfz)$ be the sample variance of $Y_i$'s in the group of patients under treatment $t$ with $\bfZ_i = \bfz$, 
$S_{t,A}^2(\bfz )$ be $S_t^2 (\bfz ) $  with $Y_i$ replaced by $Y_i - \bX_i ^{T} \hat{\bfbeta}_t (\bfz )$, 
$S_{t,B}^2(\bfz )$ be $S_t^2 (\bfz ) $  with $Y_i$ replaced by $Y_i - \bX_i ^{T} \hat{\bfbeta}(\bfz )$, and $\hat\Sigma (\bfz )$ be the sample covariance matrix of $\bX_i$'s within $\bfZ_i = \bfz$. 
It is shown in the  Supplementary Material that, under (C1)-(C3), the following estimators are consistent for $\sigma^2_U$, $\sigma^2_V$, $\sigma^2_A$, and $\sigma^2_B$, respectively, \vspace{-1mm}
\begin{align*}
\hat{\sigma}^2_U=& \,  \sum_{\bfz \in {\cal Z}}\frac{n(\bfz)}{n}\left\{ \frac{S_t^2(\bfz)}{\pi_t }+ \frac{S_s^2(\bfz)}{\pi_s}\right\}, \qquad  
\hat{\sigma}^2_V= \sum_{\bfz \in {\cal Z}}\frac{n(\bfz)}{n}\left\{ \bar{Y}_t(\bfz) -\bar{Y}_s(\bfz)\right\}^2-\hat\theta^2,   \\
\hat{\sigma}^2_A = & \,  \sum_{\bfz \in {\cal Z}}\frac{n(\bfz)}{n}\left[ \frac{S_{t,A}^2(\bfz)}{\pi_t }+ \frac{S_{s,A}^2(\bfz)}{\pi_s}
+ \{ \hat{\bfbeta}_t(\bfz) - \hat{\bfbeta}_s(\bfz)\}^{T} \hat\Sigma(\bfz) \{ \hat{\bfbeta}_t(\bfz) - \hat{\bfbeta}_s(\bfz)\} \right],\\
\hat{\sigma}^2_B = & \,  \sum_{\bfz \in {\cal Z}}\frac{n(\bfz)}{n}\left\{ \frac{S_{t,B}^2(\bfz)}{\pi_t }+ \frac{S_{s,B}^2(\bfz)}{\pi_s} \right\} , 
\end{align*}
regardless of which type of covariate-adaptive randomization method is used. 
Note that $\hat\sigma^2_U$ for $\hat\theta$ is different from the estimator defined in (36) of  
\cite{Bugni:2019aa}.


\section{Simulation Results}

There are  many publications on empirical studies under covariate-adaptive randomization in the last four decades. Some recent results are  in \cite{Senn:2010aa}, \cite{Kahan:2012aa}, and \cite{Xu:2016aa}. 

To evaluate and compare our proposed estimators $\hat\theta$, $\hat\theta_A$, and $\hat\theta_B$ in terms of estimation bias and standard deviation, and to examine variance estimators and 
the related asymptotic confidence intervals based on Theorem 1, we present some simulation results in this section. 
We consider two covariates, i.e., $\bfW = (X_1,X_2)$, where $X_1$ is binary with $\P(X_1 = 1)=1/2$ and, conditioned on $X_1$, $X_2 \sim  N(X_1- 0.5 , \, 1)$.
For the potential responses, we consider two treatments in cases I-III and three treatments in case IV.
\vspace{-1mm}
\begin{description}
	\item Case I:\ \ $\ Y^{(1)} \mid \bfW \sim N(  4X_1 + 2 X_2 , \, 1) $, $Y^{(2)} \mid \bfW \sim N(  \varphi + 4X_1 + 2 X_2 , \,1) $.
	\item Case II: $\ Y^{(1)} \mid \bfW \sim N(  4X_1 - 2 X_2 , \, 1) $, $Y^{(2)} \mid \bfW \sim N(  \varphi + 4X_1 + 2 X_2 , \,1) $. 
	\item Case III: $ Y^{(1)} \mid \bfW \sim N(  0.25 + 3X_1 + 0.2 X_2^2 , \, X_1+0.5) $, $Y^{(2)} \mid \bfW \sim N(  \varphi + 4X_1 + 2 X_2 , \,1) $. 
	\item Case IV: $ Y^{(3)}\mid W  \sim N(  \psi + 1+2X_1 -  X_2 , \,1) $, and $ Y^{(1)}$ and $ Y^{(2)} $ are the same as those in case III.
\end{description}\vspace{-1mm}
We use $\varphi =\psi =1$ in the simulation, which does not affect the relative performance of estimators and coverage probability of related confidence intervals. 

Case I has  homogeneous treatment effects; case II has treatment effect heterogeneity since the effects of $X_2$ on $Y^{(1)}$ and  $Y^{(2)}$ have different signs; case III has the most severe treatment effect heterogeneity as $Y^{(1)} \mid W$ and  $Y^{(2)} \mid W$ have very different distributions;  case IV considers multiple treatments. 

We consider 3 different $\bfZ$'s in covariate-adaptive randomization.
The first one is   $\bfZ = X_1$ with  2 levels, in which case 
the function of $\bfW$ not used in randomization but still related with the potential responses is $ h(\bfW )= X_2$. 
The second one is 
$\bfZ= (X_1, d_2)$ with  4 levels, where $d_2$ is the discretized $X_2$ with 2 categories $(-\infty , 0)$ and $[0, \infty )$, and 
$h(\bfW )$ is the continuous value of $X_2$ in $(-\infty , 0)$ or $(0, \infty)$. 
The third one is 
$\bfZ= (X_1, d_4)$ with 8 levels, where $d_4$ is the discretized $X_2$ with 4 categories $(-\infty , -0.8)$, $[-0.8, 0)$, $[0, 0.8)$, and $[0.8, \infty)$, and   
$h(\bfW )$ is the continuous value of $X_2$ in $(-\infty , -0.8)$, $(-0.8,0)$, $(0,0.8)$,  or $(0.8, \infty )$. 
In any case, $\bX = X_2$  is equivalent to $h(\bfW)$ and is used covariate adjustment. 

For the randomization method, we consider minimization with treatment allocation 1:1 or 1:2 for cases I-III, and 1:2:2 for case IV. 
Simulation results for two other randomization methods, the stratified permuted block randomization and 
the  stratified urn design can be found in the Supplementary Material. 

We consider the total sample size  $n=100$ or $500$.
For these sample sizes, the smallest possible expected numbers of patients within a stratum and treatment according to number of $Z$ levels are given in Table 3.  
It can be seen 	that when $n=100$ and $Z=(X_1,d_4)$ has 8 levels, with non-negligible probability, the number of  patients in some stratum-treatment combination is fewer than 2 and  thus calculation of estimators in (\ref{eq: theta hat})-(\ref{adj2}) and their variance estimators are not possible. Therefore, for cases I-III, we omit the scenario with  $n=100$ and $Z = (X_1,d_4)$. For  case IV, we focus on  $n=500$ and $Z = (X_1,d_2)$.

Tables 1-2 report the bias, standard deviation (SD), average estimated SD (SE), and coverage probability (CP) of asymptotic 95\% confidence interval, estimate $\pm 1.96 \,$SE, of $\hat\theta$, $\hat\theta_A$, and $\hat\theta_B$ for cases I-IV.
Every scenario is evaluated with 2,000 simulation runs. 
The following is a summary of the results in Tables 1-2.  \vspace{-2mm}
\begin{enumerate}
	\item 
	All estimators have negligible biases that are smaller than 1\% in most cases. 
	\item 
	The variance estimators or SE's are very accurate so that  the coverage probabilities of confidence intervals are adequate, 
	except for a few cases with $n=100$ and four $Z$  categories. 
	\item 
	With homogeneous  treatment effects in	case I, a more informative $\bfZ$ leads to a more efficient $\hat\theta$. However, the same phenomenon  may not exist when treatment effect heterogeneity exists, though a more informative $\bfZ$ does not lead to a less efficient $\hat\theta$.
	\item 
	Adjusting for covariates, i.e., using $\hat\theta_A$ or $\hat\theta_B$, may lead to substantial improvements over $\hat\theta$ in terms of SD, which again agrees with our theory. 
	The improvement is larger when a less informative $\bfZ$ is utilized in randomization, such as $\bfZ = X_1$.  
	Another interesting observation is that different $\bfZ$ used in randomization does not affect very much the SD of $\hat\theta_A$ or $\hat\theta_B$. 
	\item The comparison of $\hat\theta_A$ and $\hat\theta_B$ is also consistent with our theory in \S 3. Under 1:1 treatment allocation or homogeneous   treatment effects, $\hat\theta_B$ is as good as $\hat\theta_A$ and sometimes slightly better in finite sample performance. 
	When treatment allocation is 1:2 and treatment effect  heterogeneity exists, $\hat\theta_B$ is not as good as $\hat\theta_A$ and could be even worse than $\hat\theta$.  The same is observed when treatment allocation is 1:2:2. 
	\item 
	In case II with 1:1 treatment allocation, ${\rm cov} (\bX , Y^{(2)} \mid \bfZ = \bfz ) = - 
	{\rm cov} (\bX , Y^{(1)}\mid \bfZ = \bfz ) $, i.e., (\ref{equiv1}) holds and, thus, $\hat\theta$ and $\hat\theta_A$ have very similar SD's, as predicted by Theorem 2. 
	In this particular case, $\hat\theta_B$ is also as good as $\hat\theta_A$. 
	\item 
	$\hat\theta_A$ has large SD and low CP when some stratum-treatment combinations have small number of patients, such as the case of $n=100$ and $Z = (X_1 , d_2)$ with 4 levels when treatment allocation is 1:2. However, even if the smallest expected number of patients is as small as 7.7 in the case of $n=100$, 
	$Z$ having 4 levels, and 1:1 treatment allocation, $\hat\theta_A$ performs well. 
\end{enumerate}

\section{A Real Data Example}
In this section, we illustrate our methods by a real data example from \cite{Chong:2016aa} whose goal is to evaluate the contribution
of low dietary iron intake to human capital attainment by measuring the causal effect of reducing adolescent anemia on school attainment. The dataset is publicly available at \texttt{https://www.openicpsr.org/openicpsr/project/113624/version/V1/view}. In brief, \cite{Chong:2016aa} conducted an experiment on adolescents aged from 11 to 19 in rural Peru (Cajamarca) where the burden of iron deficiency is high, to study whether a low-cost intervention can encourage students to increase their iron intake and hence improve their school performance. 
A stratified permuted block randomization design was applied to assign $n= 219$ students to one of the following three promotional videos, considered as three treatments, with treatment allocation 1:1:1.  The first video shows a popular soccer player encouraging iron supplements to maximize energy; the second video shows a physician encouraging iron supplements for overall health; and the third ``placebo" video shows a dentist encouraging oral hygiene without mentioning iron at all. The strata are student's school grades $ \bfZ \in {\cal Z} =\{1,2,3,4,5\}$. 
\cite{Chong:2016aa} studied a variety of outcomes regarding cognitive function, school performance, and aspirations. As an example, we focus on the outcome of academic achievement, which is a standardized average of a student’s academic grades from the fall semester in subjects of math, foreign language, social sciences, science, and communications. The same outcome is also used by \cite{Bugni:2019aa} in an example. 

Estimates $\hat\theta$, $\hat\theta_A$, and $\hat\theta_B$ and their SE's are reported in Table 4 for the average treatment effect between the soccer player and placebo videos, or physician  and  placebo videos, together with the p-values associated with two-sided tests of no treatment effect. The estimates from $\hat\theta$ are the same as those in \cite{Bugni:2019aa}.
The covariate $X$ used in $\hat\theta_A$ and $\hat\theta_B$ is the baseline anemia status thought to have interactive effect with treatment on the outcome, as mentioned in \cite{Chong:2016aa}.
It can be seen that the SE of $\hat\theta_A$ is the smallest and, in terms of p-values,  the effect between 
physician  and  placebo videos is only marginally significant when $\hat\theta$ is used, but very significant based on $\hat\theta_A$. 

\section{Recommendations and Discussions}

To improve asymptotic efficiency,  we recommend $\hat\theta_A$ in (\ref{adj1})
since it is asymptotically better than $\hat\theta$ in (\ref{eq: theta hat}) or $\hat\theta_B$ in (\ref{adj2}). 
In the special case of two treatment arms with equal allocation, 
we recommend
$\hat\theta_B$,  since it is asymptotically equivalent to $\hat\theta_A$ and has better empirical performance. 

As full stratification according to $\bfZ$  is required, one limitation of estimators 
in  (\ref{eq: theta hat})-(\ref{adj2}) is that all strata need to have large enough sizes,  at least 10 per stratum and treatment combination \citep{Ye:2020aa}.  
Note that both covariate-adaptive randomization in treatment assignment and adjustment for covariates in estimation can  gain efficiency, and  covariate-adaptive randomization has an additional practically important advantage of balancing assignments across prognostic factors. Thus,  
how to  choose $\bfZ$ and $\bX$ is an important future research. 
It is also interesting to study estimators by combining strata of small sizes. 

\section*{Supplementary Material}
\label{SM}
Supplementary material contains all  technical proofs and more simulation results.

\clearpage
\begin{table}[h]
	\begin{center}
		\caption{Bias, standard deviation (SD), average estimated SD (SE), and coverage probability (CP) of 95\% asymptotic confidence interval under minimization for cases I-III} \vspace{-2mm}
		\resizebox{0.9\textwidth}{!}{
			\begin{tabular}{cccccccccccccc}
				\hline \\[-1.5ex]
				&      &              &                &  & \multicolumn{4}{c}{treatment   allocation 1:1} &  & \multicolumn{4}{c}{treatment   allocation 1:2} \\[0.5ex] 
				\cline{6-9} \cline{11-14} \\[-1.5ex]
				$n$ & case & $Z$          & estimator      &  & bias       & SD        & SE        & CP        &  & bias       & SD        & SE        & CP        \\[0.5ex] 
				\cline{1-4} \cline{6-9} \cline{11-14} \\[-1.5ex]
				500 & I    & $X_1$    	  & $\hat\theta$   &  & -0.0038    & 0.1980    & 0.1999    & 0.9590    &  & 0.0070     & 0.2159    & 0.2124    & 0.9465    \\
				&      &              & $\hat\theta_B$ &  & -0.0016    & 0.0909    & 0.0893    & 0.9445    &  & 0.0016     & 0.0954    & 0.0949    & 0.9510    \\
				&      &              & $\hat\theta_A$ &  & -0.0016    & 0.0908    & 0.0893    & 0.9445    &  & 0.0017     & 0.0954    & 0.0948    & 0.9490    \\[0.3ex]
				&      & $X_1$, $d_2$ & $\hat\theta$   &  & -0.0029    & 0.1492    & 0.1466    & 0.9450    &  & -0.0013    & 0.1537    & 0.1553    & 0.9560    \\
				&      &              & $\hat\theta_B$ &  & -0.0017    & 0.0900    & 0.0893    & 0.9455    &  & -0.0011    & 0.0962    & 0.0947    & 0.9440    \\
				&      &              & $\hat\theta_A$ &  & -0.0018    & 0.0901    & 0.0894    & 0.9455    &  & -0.0009    & 0.0965    & 0.0945    & 0.9435    \\[0.3ex]
				&      & $X_1$, $d_4$ & $\hat\theta$   &  & 0.0005     & 0.1143    & 0.1150    & 0.9560    &  & 0.0003     & 0.1237    & 0.1219    & 0.9485    \\
				&      &              & $\hat\theta_B$ &  & 0.0010     & 0.0903    & 0.0893    & 0.9505    &  & 0.0013     & 0.0967    & 0.0946    & 0.9425    \\
				&      &              & $\hat\theta_A$ &  & 0.0007     & 0.0908    & 0.0893    & 0.9500    &  & 0.0010     & 0.0990    & 0.0944    & 0.9415    \\[0.3ex]
				& II   & $X_1$        & $\hat\theta$   &  & 0.0085     & 0.2212    & 0.2191    & 0.9480    &  & 0.0067     & 0.2320    & 0.2303    & 0.9505    \\
				&      &              & $\hat\theta_B$ &  & 0.0086     & 0.2222    & 0.2185    & 0.9495    &  & 0.0063     & 0.2563    & 0.2541    & 0.9430    \\
				&      &              & $\hat\theta_A$ &  & 0.0086     & 0.2214    & 0.2191    & 0.9500    &  & 0.0078     & 0.2255    & 0.2212    & 0.9470    \\[0.3ex]
				&      & $X_1$, $d_2$ & $\hat\theta$   &  & 0.0084     & 0.2201    & 0.2191    & 0.9480    &  & 0.0076     & 0.2284    & 0.2251    & 0.9465    \\
				&      &              & $\hat\theta_B$ &  & 0.0076     & 0.2214    & 0.2178    & 0.9440    &  & 0.0078     & 0.2407    & 0.2344    & 0.9350    \\
				&      &              & $\hat\theta_A$ &  & 0.0085     & 0.2204    & 0.2190    & 0.9475    &  & 0.0077     & 0.2242    & 0.2212    & 0.9450    \\[0.3ex]
				&      & $X_1$, $d_4$ & $\hat\theta$   &  & 0.0057     & 0.2222    & 0.2192    & 0.9440    &  & 0.0104     & 0.2256    & 0.2230    & 0.9405    \\
				&      &              & $\hat\theta_B$ &  & 0.0061     & 0.2233    & 0.2177    & 0.9425    &  & 0.0108     & 0.2289    & 0.2254    & 0.9420    \\
				&      &              & $\hat\theta_A$ &  & 0.0057     & 0.2221    & 0.2190    & 0.9430    &  & 0.0101     & 0.2252    & 0.2211    & 0.9420    \\[0.3ex]
				& III  & $X_1$        & $\hat\theta$   &  & 0.0003     & 0.1716    & 0.1731    & 0.9475    &  & 0.0072     & 0.1691    & 0.1667    & 0.9425    \\
				&      &              & $\hat\theta_B$ &  & 0.0029     & 0.1495    & 0.1477    & 0.9500    &  & 0.0048     & 0.1675    & 0.1656    & 0.9475    \\
				&      &              & $\hat\theta_A$ &  & 0.0031     & 0.1496    & 0.1479    & 0.9480    &  & 0.0081     & 0.1546    & 0.1533    & 0.9465    \\[0.3ex]
				&      & $X_1$, $d_2$ & $\hat\theta$   &  & 0.0016     & 0.1580    & 0.1593    & 0.9490    &  & 0.0032     & 0.1603    & 0.1595    & 0.9465    \\
				&      &              & $\hat\theta_B$ &  & 0.0050     & 0.1468    & 0.1470    & 0.9460    &  & 0.0061     & 0.1621    & 0.1576    & 0.9420    \\
				&      &              & $\hat\theta_A$ &  & 0.0032     & 0.1464    & 0.1474    & 0.9480    &  & 0.0052     & 0.1559    & 0.1523    & 0.9440    \\[0.3ex]
				&      & $X_1$, $d_4$ & $\hat\theta$   &  & 0.0042     & 0.1528    & 0.1525    & 0.9465    &  & 0.0049     & 0.1601    & 0.1557    & 0.9405    \\
				&      &              & $\hat\theta_B$ &  & 0.0080     & 0.1495    & 0.1468    & 0.9425    &  & 0.0094     & 0.1577    & 0.1539    & 0.9395    \\
				&      &              & $\hat\theta_A$ &  & 0.0048     & 0.1493    & 0.1474    & 0.9450    &  & 0.0069     & 0.1573    & 0.1521    & 0.9410    \\
				\\[-1.5ex]
				100 & I    & $X_1$    & $\hat\theta$   &  & 0.0111     & 0.4441    & 0.4486    & 0.9545    &  & 0.0066     & 0.4609    & 0.4747    & 0.9580    \\
				&      &              & $\hat\theta_B$ &  & 0.0038     & 0.2035    & 0.1982    & 0.9425    &  & -0.0017    & 0.2114    & 0.2092    & 0.9490    \\
				&      &              & $\hat\theta_A$ &  & 0.0038     & 0.2035    & 0.1983    & 0.9425    &  & -0.0014    & 0.2126    & 0.2083    & 0.9485    \\[0.3ex]
				&      & $X_1$, $d_2$ & $\hat\theta$   &  & 0.0108     & 0.3316    & 0.3305    & 0.9425    &  & 0.0034     & 0.3475    & 0.3501    & 0.9490    \\
				&      &              & $\hat\theta_B$ &  & 0.0029     & 0.2021    & 0.1981    & 0.9465    &  & 0.0016     & 0.2181    & 0.2092    & 0.9365    \\
				&      &              & $\hat\theta_A$ &  & 0.0013     & 0.2148    & 0.2013    & 0.9440    &  & -0.0035    & 0.4016    & 0.2190    & 0.9140    \\[0.3ex]
				& II   & $X_1$        & $\hat\theta$   &  & -0.0057    & 0.4857    & 0.4902    & 0.9515    &  & -0.0143    & 0.5131    & 0.5157    & 0.9445    \\
				&      &              & $\hat\theta_B$ &  & -0.0044    & 0.4958    & 0.4828    & 0.9420    &  & -0.0152    & 0.5710    & 0.5612    & 0.9405    \\
				&      &              & $\hat\theta_A$ &  & -0.0049    & 0.4878    & 0.4894    & 0.9500    &  & -0.0101    & 0.4965    & 0.4943    & 0.9455    \\[0.3ex]
				&      & $X_1$, $d_2$ & $\hat\theta$   &  & -0.0062    & 0.4874    & 0.4904    & 0.9495    &  & -0.0102    & 0.5013    & 0.5031    & 0.9545    \\
				&      &              & $\hat\theta_B$ &  & -0.0077    & 0.4981    & 0.4774    & 0.9355    &  & -0.0096    & 0.5400    & 0.5114    & 0.9305    \\
				&      &              & $\hat\theta_A$ &  & -0.0076    & 0.4906    & 0.4904    & 0.9455    &  & -0.0147    & 0.5993    & 0.4999    & 0.9390    \\[0.3ex]
				& III  & $X_1$        & $\hat\theta$   &  & 0.0003     & 0.3855    & 0.3869    & 0.9475    &  & -0.0074    & 0.3729    & 0.3717    & 0.9470    \\
				&      &              & $\hat\theta_B$ &  & 0.0050     & 0.3349    & 0.3266    & 0.9410    &  & -0.0052    & 0.3744    & 0.3649    & 0.9365    \\
				&      &              & $\hat\theta_A$ &  & 0.0047     & 0.3314    & 0.3291    & 0.9465    &  & 0.0045     & 0.3486    & 0.3390    & 0.9400    \\[0.3ex]
				&      & $X_1$, $d_2$ & $\hat\theta$   &  & -0.0004    & 0.3541    & 0.3566    & 0.9445    &  & -0.0054    & 0.3572    & 0.3569    & 0.9510    \\
				&      &              & $\hat\theta_B$ &  & 0.0104     & 0.3313    & 0.3236    & 0.9425    &  & 0.0084     & 0.3663    & 0.3445    & 0.9305    \\
				&      &              & $\hat\theta_A$ &  & 0.0015     & 0.3315    & 0.3298    & 0.9435    &  & 0.0003     & 0.5643    & 0.3456    & 0.9310    \\[0.5ex]
				\hline
		\end{tabular}}
	\end{center}
\end{table}

\begin{table}[ht!]
	\begin{center}
		\captionsetup{justification=centering}
		\caption{Bias, standard deviation (SD), average estimated SD (SE), and coverage probability (CP) of 95\% asymptotic confidence interval under minimization for case IV with $n=500$}
		\vspace{-2mm}
		\resizebox{!}{!}{
			\begin{tabular}{ccccccccc}
				\hline \\[-1.5ex]
				$t$ & $s$ & $\theta$ & estimator      &  & bias    & SD     & SE     & CP     \\[0.5ex]
				\cline{1-4} \cline{6-9} \\[-1.5ex] 
				2 & 1 & 1                    & $\hat\theta$   &  & -0.0007 & 0.1907 & 0.1901 & 0.9515 \\
				&   &                      & $\hat\theta_B$ &  & 0.0040  & 0.1840 & 0.1821 & 0.9470 \\
				&   &                      & $\hat\theta_A$ &  & 0.0058  & 0.1777 & 0.1726 & 0.9375 \\
				&   &                      &                &  &         &        &        &        \\[-1.5ex]
				3 & 1 & 1                    & $\hat\theta$   &  & -0.0004 & 0.1541 & 0.1546 & 0.9445 \\
				&   &                      & $\hat\theta_B$ &  & 0.0037  & 0.1616 & 0.1615 & 0.9445 \\
				&   &                      & $\hat\theta_A$ &  & 0.0052  & 0.1479 & 0.1460 & 0.9395 \\
				&   &                      &                &  &         &        &        &        \\[-1.5ex]
				3 & 2 & 0                    & $\hat\theta$   &  & 0.0004  & 0.2094 & 0.2082 & 0.9505 \\
				&   &                      & $\hat\theta_B$ &  & -0.0002 & 0.2077 & 0.2048 & 0.9445 \\
				&   &                      & $\hat\theta_A$ &  & -0.0006 & 0.2019 & 0.2007 & 0.9495 \\[0.5ex]
				\hline
		\end{tabular}}
		
		\vspace{10mm}
		\caption{Smallest expected number of patients  among all stratum-treatment combinations}
		\vspace{-2mm}
		\resizebox{!}{!}{
			\begin{tabular}{ccccccc}
				\hline \\[-1.5ex]
				$n$   & $Z$        & number of levels &  & 1:1 allocation & 1:2 allocation & 1:2:2 allocation \\[0.5ex]
				\cline{1-3} \cline{5-7} \\[-1.5ex]
				100 & $X_1$      & 2                &  & 25.0           & 16.7           & 10.0             \\
				& $X_1, d_2$ & 4                &  & \ 7.7            &\ \ 5.1            & \ \ 3.1              \\
				& $X_1, d_4$ & 8                &  & \ 2.4            & \ 1.6            & \ 1.0              \\
				&            &                  &  &                &                &                  \\[-1.5ex]
				500 & $X_1$      & 2                &  & \ 125            & 83.3           & \ \ 50               \\
				& $X_1, d_2$ & 4                &  & 38.6           & 25.7           & 15.4             \\
				& $X_1, d_4$ & 8                &  & 12.1           & \ \ 8.1            & \ \ 4.8              \\[0.5ex]
				\hline
		\end{tabular}}
		
		\vspace{10mm}
		\caption{Results from real data analysis}
		\vspace{-2mm}
		\resizebox{!}{!}{
			\begin{tabular}{ccccccccc}
				\hline \\[-1.5ex]
				&  & \multicolumn{3}{c}{soccer versus   placebo}         &  & \multicolumn{3}{c}{physician versus  placebo} \\[0.5ex]
				\cline{3-5} \cline{7-9} \\[-1.5ex]
				estimator 	   &  & estimate         & SE          & p-value      &  & estimate          & SE           & p-value \\[0.5ex]
				\hline \\[-1.5ex]
				$\hat\theta$   &  & -0.051       & 0.205       & 0.803       &  & 0.409         & 0.207        & 0.048 \\
				$\hat\theta_B$ &  & -0.089       & 0.203       & 0.661       &  & 0.444         & 0.202        & 0.028 \\
				$\hat\theta_A$ &  & -0.048       & 0.199       & 0.807       &  & 0.480         & 0.198        & 0.015 \\[0.5ex]
				\hline
		\end{tabular}}
	\end{center}
\end{table}
\clearpage 

\singlespacing
\bibliographystyle{apalike}
\bibliography{reference}
\clearpage
\pagenumbering{arabic}
\begin{center}
	{\LARGE\bf Supplementary Material}
\end{center}
\medskip
\setcounter{equation}{0}
\setcounter{table}{0}
\renewcommand{\theequation}{S\arabic{equation}}
\renewcommand{\thetable}{S\arabic{table}}
\section*{Technical Proofs}
\subsection*{Proof of Theorem 1}
{\bf Asymptotics for $ \hat\theta $}: We start with deriving the asymptotic distribution of $ \hat\theta $. As mentioned in the main article, the key in this proof is decomposing $ \hat\theta- \theta $ as the sum of $ U $ and $ V $, where 
\begin{align}
U= &\sum_{\bfz \in {\cal Z}} \frac{n(\bfz)}{n}\left[ \{ \bar{Y}_t(\bfz)-\bar{Y}_s(\bfz)\} -\{ E(Y^{(t)}\mid\bfZ=\bfz) -E(Y^{(s)}\mid \bfZ=\bfz) \}\right],\label{suppeq: U} \\
V= &\sum_{\bfz \in {\cal Z}} \frac{n(\bfz)}{n} \{ E(Y^{(t)}\mid \bfZ=\bfz) -E(Y^{(s)}\mid\bfZ=\bfz) \} -\theta. \label{suppeq: V}
\end{align}

Let $\calA = (\bfZ_1, \ldots, \bfZ_n, I_1, \ldots, I_n)$, $\mI (A)$ be the indicator function of event $A$, and $n_t(\bfz)= \sum_{i=1}^n \mI (I_i=e_t, \bfZ_i=\bfz)$
be the number of patients within stratum $\bfZ_i=\bfz$ and  treatment $t$.
Under (C1)-(C3), $n(\bfz ), n_t(\bfz)$, and $n_s(\bfz)$ are  functions of  $\calA$ and the conditional expectation  
$ E \{ \bar{Y}_t(\bfz) \mid \calA \}$ is equal to 
$	E(Y^{(t)}\mid\bfZ=\bfz)$,  which implies that 
$E(U\mid\calA)=0$ a.s. and $E(U)=0$.
Note that
\begin{align}
V&= \frac{1}{n} \sum_{\bfz \in {\cal Z}} \sum_{i=1}^n \left[\mI (\bfZ_i=\bfz) \{ E(Y^{(t)}\mid\bfZ=\bfz) -E(Y^{(s)}\mid\bfZ=\bfz) \} \right]-\theta\nonumber\\
&= \frac{1}{n}\sum_{i=1}^n\left[  \sum_{\bfz \in {\cal Z}} \mI (\bfZ_i=\bfz) \{ E(Y^{(t)}\mid\bfZ=\bfz) -E(Y^{(s)}\mid\bfZ=\bfz) \}\right] -\theta\nonumber\\
&= \frac{1}{n} \sum_{i=1}^n  E(Y_i^{(t)}-Y_i^{(s)}\mid\bfZ_i) -\theta \nonumber 
\end{align}
Hence, $E(V)=0$ and consequently  $E(\hat{\theta})=\theta$, which establishes the unbiasedness of 
$\hat{\theta}$. 

Next, we establish the asymptotic normality of $ \hat\theta $.
From the previous derivations of $V$,  $V$ is simply an average of independent and identically distributed terms and from the Central Limit Theorem,
$$
\sqrt{n} V\xrightarrow{d} N(0,\sigma_V^2). $$
We now turn to $U$. 	
Notice that  $I_i$'s are involved in $U$, which  results in complicated dependence and is the major difficulty in deriving the asymptotic distribution of $U$. A useful technique that can largely simplify the problem is to derive the  distribution of $U$ conditional on  $\calA = (\bfZ_1, \ldots, \bfZ_n, I_1, \ldots, I_n)$. We define  
$$U(\bfz) = \frac{n(\bfz)}{n}\left[ \{ \bar{Y}_t(\bfz)-\bar{Y}_s(\bfz)\} -\{ E(Y^{(t)}\mid\bfZ=\bfz) -E(Y^{(s)}\mid\bfZ=\bfz) \}\right]$$
so that $U=\sum_{\bfz \in {\cal Z}} U(\bfz)$. 
From the fact that $E\{U(\bfz)\mid\calA\} =0 $ a.s.\ and Lindeberg's Central Limit Theorem, we conclude that, for every $\bfz$, as $n\rightarrow \infty$,
\begin{align}
\frac{ U(\bfz)}{\sqrt{\var \{ U(\bfz)\mid\calA\}}} \, \mid \, \calA \xrightarrow{d} N(0,1). \nonumber
\end{align}
The Lindeberg's condition will be verified at the end. 

Next, we prove that 
$
{\rm cov} \{  \bar{Y}_t(\bfz) -\bar{Y}_s(\bfz), \bar{Y}_t(\bfz') -\bar{Y}_s(\bfz')\mid\calA\}=0 $ a.s.\ for   $\bfz\neq \bfz'$. 
By definition this conditional covariance is equal to 

\begin{align*}
& {\rm cov} \left[ \sum_{i=1}^n\mI (\bfZ_i=\bfz) \left\{ \frac{\mI(I_i=e_t) Y_{i}^{(t)}}{n_t(\bfz)}- \frac{\mI(I_i=e_s) Y_{i}^{(s)}}{n_s(\bfz)} \right\} ,\right.  \\
&\qquad \qquad\qquad\qquad\qquad\qquad\qquad\qquad\left. \sum_{j=1}^n\mI (\bfZ_j=\bfz') \left\{ \frac{\mI(I_j=e_t) Y_{j}^{(t)}}{n_t(\bfz')}- \frac{\mI(I_j=e_s) Y_{j}^{(s)}}{n_s(\bfz')} \right\}\mid\calA\right]\\
&=  \, \sum_{i=1}^n \sum_{j=1}^n \mI (\bfZ_i=\bfz)\mI (\bfZ_j=\bfz') \ {\rm cov}\left\{\frac{\mI(I_i=e_t) Y_{i}^{(t)}}{n_t(\bfz)}- \frac{\mI(I_i=e_s) Y_{i}^{(s)}}{n_s(\bfz)},  \right.\\
&\qquad \qquad\qquad\qquad\qquad\qquad\qquad\qquad\left.\frac{\mI(I_j=e_t) Y_{j}^{(t)}}{n_t(\bfz')}- \frac{\mI(I_j=e_s) Y_{j}^{(s)}}{n_s(\bfz')} \mid\calA\right\}
\end{align*}
because $\{ \bfZ_i, i=1,\ldots,n \} \subset \calA$. 
When $i=j$, $\mI (\bfZ_i=\bfz)\mI (\bfZ_i=\bfz') = 0$ for $\bfz\neq\bfz'$. 
When $i\neq j$, the terms are also equal to zero because ${\rm cov} (Y_{i}^{(\ell)}, Y_{j}^{(m)}\mid\calA)=0$ a.s. for $\ell , m=1,\dots, k$ from (C1)-(C2). From the definition of $U(\bfz )$, this also proves that  ${\rm cov} \{ U(\bfz), U(\bfz')\mid\calA \}=0$ a.s.\ for $\bfz\neq \bfz'$.
Then,  it follows from the delta method that 
\[
\frac{U}{ \sqrt{\sum_{\bfz \in {\cal Z}}	\var \{ U(\bfz)\mid\calA\} }} \, \mid \, \calA \xrightarrow{d} N(0,1).
\]
From the bounded convergence theorem, this result still holds unconditionally, i.e., 
\[
\frac{U}{ \sqrt{\sum_{\bfz \in {\cal Z}}	\var \{ U(\bfz)\mid\calA\} }}  \xrightarrow{d} N(0,1).
\]
Note that 
\begin{align}
&n	\sum_{\bfz \in {\cal Z} } \var\{U(\bfz)\mid\calA\}  \nonumber\\
&= \sum_{\bfz \in {\cal Z}} \frac{n^2(\bfz)}{n} \var \left\{  \frac{1 }{n_t(\bfz)} \sum_{i: \bfZ_i=\bfz} \mI(I_i=e_t) Y_i^{(t)} - \frac{1}{n_s(\bfz)}\sum_{i: \bfZ_i=\bfz} \mI(I_i=e_s)Y_i^{(s)}   \mid \calA \right\}\nonumber\\
&= \sum_{\bfz \in {\cal Z}} \frac{n^2(\bfz)}{n} \var \left\{  \sum_{i: \bfZ_i=\bfz}  \frac{ \mI(I_i=e_t) Y_i^{(t)}}{n_t(\bfz)} - \frac{ \mI(I_i=e_s) Y_i^{(s)} }{n_s(\bfz)} \mid \calA \right\}\nonumber\\
&=  \sum_{\bfz \in {\cal Z}} \frac{n^2(\bfz)}{n} \sum_{i: \bfZ_i=\bfz} \var\left\{  \frac{ \mI(I_i=e_t) Y_i^{(t)} }{n_t(\bfz)} - \frac{ \mI(I_i=e_s) Y_i^{(s)} }{n_s(\bfz)}  \mid \calA \right\} \nonumber\\
&=  \sum_{\bfz \in {\cal Z}} \frac{n^2(\bfz)}{n} \sum_{i: \bfZ_i=\bfz} \left\{  \frac{ \mI(I_i=e_t) \var(Y_i^{(t)}\mid\bfZ_i) }{n^2_t(\bfz)} + \frac{ \mI(I_i=e_s) \var(Y_i^{(s)}\mid\bfZ_i) }{n^2_s(\bfz)} \right\} \nonumber\\
&= \sum_{\bfz \in {\cal Z}} \frac{n(\bfz)}{n} \left\{ \frac{n(\bfz)}{n_t(\bfz)} \var(Y^{(t)}\mid\bfZ=\bfz)+ \frac{n(\bfz)}{n_s(\bfz)} \var(Y^{(s)}\mid\bfZ=\bfz)\right\} \nonumber \\
&=\sum_{\bfz \in {\cal Z}} \P(\bfZ=\bfz)\left\{ \frac{1}{\pi_t} \var(Y^{(t)}\mid\bfZ=\bfz)+ \frac{1}{\pi_s}\var(Y^{(s)}\mid\bfZ=\bfz)\right\} +o_p(1)\nonumber \\
&= E\left\{ \frac{1}{\pi_t} \var(Y^{(t)}\mid\bfZ) +\frac{1}{\pi_s} \var(Y^{(s)}\mid\bfZ)\right\}+o_p(1)\nonumber \\
& = \sigma^2_U +o_p(1)\nonumber 
\end{align}
where the first equality is because $n(\bfz)$ is a function of $\calA$, the third equality is because  the summands are mutually independent conditional on $\calA$, the fourth equality is from (C1)-(C2), the sixth equality is because $n(\bfz)/n=\P(\bfZ=\bfz)+o_p(1)$, and from (C2)-(C3), $n_t(\bfz)/n(\bfz)=n(\bfz)^{-1}D_t(\bfz)+\pi_t=\pi_t+o_p(1)$ for every $ t $.
From Slutsky's theorem, this proves that as $n\rightarrow \infty$,
$$ \sqrt{n} U / \sigma_U  \xrightarrow{d} N(0,1).$$
Combining  the results for $U$ and $V$ and using the fact that  $U$ and $V$  are uncorrelated because $E(U V)=E\{VE(U\mid\calA)\}=0$, we conclude that the results for $ \hat\theta $ hold.

For completeness, we verify the Lindeberg's condition, by rearranging $ U(z) $ as
\begin{align}
U(z)=  \sum_{i=1}^n (K_{i}^{(t)}-K_{i}^{(s)}), \qquad K_i^{(t)}= \frac{n(z)}{n} \frac{\mI (I_i=e_t, Z_i=z)}{n_t(z)} \{Y_i^{(t)}- E(Y^{(t)}\mid Z_i=z)\} \nonumber
\end{align}
We have already shown that $ E (K_i^{(t)}- K_i^{(s)}\mid {\cal D})=0$ and 
\begin{align}
&\var (K_i^{(t)}- K_i^{(s)}\mid {\cal D}) \nonumber\\
&= \frac{n^2(z)}{n^2n_t^2(z)} \mI (I_i=e_t, Z_i=z) \var(Y^{(t)}\mid Z=z) +  \frac{n^2(z)}{n^2n_s^2(z)} \mI (I_i=e_s, Z_i=z) \var(Y^{(s)}\mid Z=z).\nonumber
\end{align} 
Then, the Lindeberg's condition holds because for any $ \epsilon>0 $, 
\begin{align*}
&\sum_{i=1}^n E\left[  \frac{(K_i^{(t)}- K_i^{(s)})^2}{\var(U(z)\mid {\cal D})}\mI \left\{\frac{( K_i^{(t)}- K_i^{(s)})^2}{  \var (U(z)\mid {\cal D}) }>\epsilon\right\}\mid {\cal D}\right] \\
&= \sum_{i=1}^n  \frac{\var (K_i^{(t)}- K_i^{(s)}\mid {\cal D}) }{\var(U(z)\mid {\cal D})} E\left[\frac{(K_i^{(t)}- K_i^{(s)})^2}{\var (K_i^{(t)}- K_i^{(s)}\mid {\cal D})}\mI \left\{\frac{( K_i^{(t)}- K_i^{(s)})^2}{  \var (U(z)\mid {\cal D}) }>\epsilon\right\}\mid {\cal D}\right]\\
&\leq \max_i E\left[\frac{(K_i^{(t)}- K_i^{(s)})^2}{\var (K_i^{(t)}- K_i^{(s)}\mid {\cal D})}\mI \left\{\frac{( K_i^{(t)}- K_i^{(s)})^2}{\var (K_i^{(t)}- K_i^{(s)}\mid {\cal D})}>\epsilon \frac{\var (U(z)\mid \cal D)}{\var (K_i^{(t)}- K_i^{(s)}\mid {\cal D})}\right\}\mid {\cal D}\right]\\
&=o(1)
\end{align*}
where the third line is because $ \sum_{i=1}^n \var (K_i^{(t)}- K_i^{(s)}\mid {\cal D})= \var(U(z)\mid \cal D) $, and the last line is because $  K_i^{(t)}- K_i^{(s)}/\sqrt{ \var (K_i^{(t)}- K_i^{(s)}\mid {\cal D})  } $ has zero expectation  and unit variance, and that $ \max_i\var (K_i^{(t)}- K_i^{(s)}\mid {\cal D}) / \var (U(z)\mid {\cal D})\leq \max( \{ n_t(z)\}^{-1}, \{n_s(z)\}^{-1})=o(1)$. \\

{\bf Asymptotics for $ \hat\theta_A, \hat\theta_B$}: To establish the results for $ \hat\theta_A $ and $ \hat\theta_B $,  we need the following lemma about the asymptotic limits  of $\hat{\bfbeta}_t(\bfz)$ used in $\hat\theta_A$  and  $\hat{\bfbeta} (\bfz)$ used in $\hat\theta_B$. 

\begin{lemma} \label{lemma: coef}
	Under the conditions of Theorem \ref{theo1}, $\hat{\bfbeta}_t(\bfz) = \bfbeta_t(\bfz)+o_p(1)$ and $\hat{\bfbeta} (\bfz)=\bfbeta(\bfz)+o_p(1)$, for every $\bfz \in {\cal Z}$, where $\bfbeta_t(\bfz )$ and $\bfbeta (\bfz )$ are defined in Theorem \ref{theo1}.  
\end{lemma}

\begin{proof}[of Lemma~\ref{lemma: coef}] We prove the result for $\hat{\bfbeta}_t(\bfz)$. 
	The proof for  $\hat{\bfbeta}(\bfz)$ is analogous and  omitted. 
	The numerator of $\hat{\bfbeta}_t(\bfz)$ equals
	\[	
	\sum_{i=1}^n  \mI (I_i=e_t, \bfZ_i=\bfz) \bX_i Y_i-\frac{1}{n_t(\bfz)} \sum_{i=1}^n \mI (I_i=e_t, \bfZ_i=\bfz) \bX_i  \sum_{i=1}^n  \mI (I_i=e_t, \bfZ_i=\bfz)  Y_i
	\]
	Conditional on $\calA$, the first term is an average of independent  random variables. Assuming (C2) and existence of the second moment of $\bX Y^{(t)}$, by the weak law of large numbers for independent random variables, we conclude that, for any $\epsilon>0$,
	\[
	\lim_{n\rightarrow \infty} \P\left\{ \frac{1}{n} \mid \sum_{i=1}^n  \mI (I_i=e_t, \bfZ_i=\bfz) \bX_i Y_i- n_t(\bfz)E(\bX Y^{(t)}\mid\bfZ=\bfz) \mid   \geq  \epsilon \, \mid \, \calA\right\}=0,
	\]
	which together with the bounded convergence theorem and $n_t(\bfz)/n=\pi_t \P(\bfZ=\bfz)+o_p(1)$ by (C3) 
	implies that
	\[
	\frac{1}{n_t(\bfz )} \sum_{i=1}^n  \mI (I_i=e_t, \bfZ_i=\bfz) \bX_i Y_i
	=E(\bX Y^{(t)}\mid\bfZ=\bfz) +o_p(1).
	\]
	Similarly, we can show the result with $\bX_iY_i$ replaced by $\bX_i$ or $Y_i$ and, therefore, 
	\begin{align*}
	& \	\frac{1}{n_t(\bfz )} \sum_{i=1}^n   \mI (I_i=e_t, \bfZ_i=\bfz) \bX_i Y_i-
	\frac{1}{n_t^2(\bfz )} \sum_{i=1}^n \mI (I_i=e_t, \bfZ_i=\bfz) \bX_i  \sum_{i=1}^n  \mI (I_i=e_t, \bfZ_i=\bfz)  Y_i	\\
	=& \ \{ E(\bX Y^{(t)}\mid\bfZ=\bfz) + o_p(1) \}- \{E(\bX \mid\bfZ=\bfz) +o_p(1)\} \{E( Y^{(t)} \mid \bfZ = \bfz ) +o_p(1) \}\\
	=&  \ {\rm cov} ( \bX, Y^{(t)} \mid\bfZ=\bfz) + o_p(1) . 
	\end{align*}
	The denominator of $\hat{\bfbeta}_t(\bfz)$ can be treated similarly, which leads to  
	$$
	\frac{1}{n_t(\bfz)} \sum_{ i: I_i =e_t, \bfZ_i = \bfz } \{\bX_i - \bar{\bX}_t(\bfz )\}\{\bX_i - \bar{\bX}_t(\bfz )\}^T = \var(\bX\mid\bfZ=\bfz)+o_p(1).
	$$
	The proof is completed by using the definition of $ \bfbeta_t(\bfz)$. 
\end{proof}

Next, consider $\hat{\theta}_B$. Let 
$$ U_B = \sum_{\bfz \in {\cal Z}} \frac{n(\bfz)}{n} \left[ \bar{Y}_t(\bfz) -\bar{Y}_s(\bfz) - \{\bar{\bX}_t(\bfz) -\bar{\bX}_s(\bfz)\}^T  \bfbeta (\bfz) -\{ E(Y^{(t)}\mid\bfZ=\bfz) -E(Y^{(s)}\mid\bfZ=\bfz) \}\right]$$
Similar to the decomposition for $\hat\theta - \theta$ in (\ref{suppeq: U})-(\ref{suppeq: V}), we have the following  decomposition for $ \hat\theta_B- \theta $:
$$
\hat\theta_B - \theta = U_B + V 
- \sum_{\bfz \in {\cal Z}} \frac{n(\bfz)}{n} \{\bar{\bX}_1(\bfz) -\bar{\bX}_0(\bfz)\}^T 
\{\hat{\bfbeta} (\bfz ) - \bfbeta (\bfz ) \}
$$
The last term is $o_p(n^{-1/2})$ because $\hat{\bfbeta} (\bfz) -\bfbeta(\bfz)=o_p(1)$ by Lemma 1 and $\bar{\bX}_1(\bfz) -\bar{\bX}_0(\bfz)=O_p(n^{-1/2})$ that can be shown using a similar technique in  the proof of  $ \hat\theta $. 
To derive the asymptotic distribution of  $U_B$, we apply the same techniques used to treat $U$ in the proof of $ \hat\theta $, i.e., conditioned on $\calA$, $U_B$ can be shown to be an average of independent  terms so that conditioned on $\calA$, $U_B$ is asymptotically normal and, therefore, unconditionally it is also asymptotically normal. Since $E(U_B\mid\calA)=0$, it remains to find the conditional variance of $U_B$ given $\calA$, which is 
\begin{align}
\var(\sqrt{n}U_B\mid\calA) &= \sum_{\bfz\in {\cal Z}} \frac{n(\bfz)}{n} \! \!\left[  \frac{n(\bfz)}{n_t(\bfz)} \var\{Y^{(t)}- \bX^T \! \bfbeta (\bfz) \mid\bfZ=\bfz\}\!+\! \frac{n(\bfz)}{n_s(\bfz)} \var\{Y^{(s)}- \bX^T \! \bfbeta (\bfz)  \mid\bfZ=\bfz\}\!\right] \nonumber \\
&= E\left[ \frac{1}{\pi_t} \var\{Y^{(t)} - \bX^T \bfbeta (\bfZ)\mid\bfZ\}+\frac{1}{\pi_s} \var\{Y^{(s)} - \bX^T \bfbeta (\bfZ)\mid\bfZ\}\right]+o_p(1)\nonumber \\
& = \sigma_B^2 + o_p(1)\nonumber
\end{align}
This completes the proof of $ \hat\theta_B $.

Next, we consider $\hat{\theta}_A$. Define 
$$
U_A =  \sum_{\bfz \in {\cal Z}} \frac{n(\bfz)}{n} \sum_{\ell=t,s} (-1)^{\ell=s} \left[ 
\bar{Y}_\ell(\bfz) -   \{\bar{\bX}_\ell(\bfz) -\bar{\bX}(\bfz)\}^T \bfbeta_\ell(\bfz) - E(Y^{(\ell)} \mid\bfZ=\bfz)  \right] 
$$
Then, 
$$
\hat\theta_A - \theta = U_A + V - \sum_{\bfz \in {\cal Z}} \frac{n(\bfz)}{n}\sum_{\ell=t, s} (-1)^{\ell=s} \{(\bar{\bX}_\ell(\bfz) -\bar{\bX}(\bfz)\}^T  \{\hat{\bfbeta}_\ell(\bfz)-\bfbeta_\ell(\bfz)\}  
$$
where the last term is  $o_p(n^{-1/2})$ because $\hat{\bfbeta}_t(\bfz) -\bfbeta_t(\bfz)=o_p(1)$ by Lemma \ref{lemma: coef} and $\bar{\bX}_j(\bfz)- \bar{\bX}(\bfz)=O_p(n^{-1/2})$. It remains to derive the asymptotic distribution of $U_A$. Consider a further decomposition 
$$ U_A = U_{A1} + U_{A2}$$
where
\begin{align*}
U_{A1} = & \sum_{\bfz \in {\cal Z}} \frac{n(\bfz)}{n} \sum_{\ell=t,s} (-1)^{\ell=s} 
[ \bar{Y}_\ell(\bfz)-   E(Y^{(\ell)}\mid\bfZ=\bfz) -  \{\bar{\bX}_\ell(\bfz) - E(\bX\mid\bfZ=\bfz)\}^T \bfbeta_\ell (\bfz )]  \\
U_{A2} = & \sum_{\bfz \in {\cal Z}} \frac{n(\bfz)}{n} \{\bar{\bX}(\bfz)- E(\bX\mid\bfZ=\bfz)\}^T   \left\{  \bfbeta_t(\bfz) -\bfbeta_s(\bfz)\right\}
\end{align*}
and ${\rm cov} (U_{A1}, U_{A2}\mid\calA)= E(U_{A1} U_{A2}\mid\calA) = 0$, which can be seen from ${\rm cov} \{Y_{i}^{(t)}-\bX_i\bfbeta_t(\bfz) , \bX_j\mid\bfZ_i=\bfz, \bfZ_j\}=0$ a.s. for any $j, i=1, \ldots, n$, and $ t=1,\ldots, k $.
The asymptotic normality of $U_{A1}$ can be derived in the same way as that for $U_B$ with $E(U_{A1}\mid \calA )=0$ and 
$$
\var(\sqrt{n}U_{A1}\mid\calA) = E\left[ \frac{1}{\pi_t} \var\{Y^{(t)} -  \bX^T \bfbeta_t(\bfz)\mid\bfZ\} +\frac{1}{\pi_s} \var\{Y^{(s)}- \bX^T \bfbeta_s(\bfz) \mid\bfZ\}\right]+o_p(1).$$
Note that $U_{A2}$ is an average of independent and identically distributed terms and its asymptotic normality follows directly from the Central Limit Theorem. Specifically,
$\sqrt{n}U_{A2}$ converges in distribution to normal with mean 0 and variance 
$E\left[ \{\bfbeta_t(\bfZ) - \bfbeta_s(\bfZ)\}^T \var(\bX\mid\bfZ)\{\bfbeta_t(\bfZ) - \bfbeta_s(\bfZ)\}\right].$ This proves the first result for $ \hat\theta_A $
because the sum of the asymptotic variances of $U_{A1}$ and $U_{A2}$ is exactly $\sigma^2_A$.  The proof is completed.

\subsection*{Proof of Theorem \ref{theo2}}
Define $\Sigma (\bfz ) = {\rm var} ( \bX \mid \bfZ= \bfz )$, 
\begin{align*}
\sigma_A^2(\bfz ) = & \frac{\var (Y^{(t)}- \bX^T \! \bfbeta_t(\bfz )\mid\bfZ=\bfz)}{\pi_t} + \frac{\var (Y^{(s)}- \bX^T \! \bfbeta_s(\bfz )\mid\bfZ=\bfz)}{\pi_s} \\
& + \{\beta_t(\bfz)- \bfbeta_s(\bfz)\}^T  \Sigma(\bfz)\{\bfbeta_t(\bfz)- \bfbeta_s(\bfz)\}
\\
\sigma_U^2 (\bfz ) = & \frac{\var (Y^{(t)}\mid\bfZ=\bfz)}{\pi_t} + \frac{\var (Y^{(s)}\mid\bfZ=\bfz)}{\pi_s}  
\end{align*}
Then
$\sigma^2_A = \sum_{\bfz \in {\cal Z}} \P(\bfZ = \bfz ) \sigma_A^2(\bfz )$ and $ \sigma^2_U = \sum_{\bfz \in {\cal Z}} \P(\bfZ = \bfz ) \sigma_U^2(\bfz ) $.
Hence, the result for $ \sigma_U^2- \sigma_A^2 $ follows from 
\begin{align*}
&\sigma_A^2(\bfz ) - \sigma_U^2 (\bfz ) \nonumber\\
= & \, \frac{\{\bfbeta_t(\bfz)\}^T \Sigma(\bfz) \bfbeta_t(\bfz) -2{\rm cov}(  \bX^T \! \bfbeta_t(\bfz), Y^{(t)}\mid\bfZ=\bfz) }{\pi_t}\nonumber\\
& + \frac{ \{\bfbeta_s(\bfz)\}^T \Sigma(\bfz) \bfbeta_s(\bfz) -2{\rm cov}(  \bX^T \! \bfbeta_s(\bfz), Y^{(s)}\mid\bfZ=\bfz)}{\pi_s}+ \{ \bfbeta_t(\bfz)- \bfbeta_s(\bfz)\}^T \Sigma(\bfz)\{\bfbeta_t(\bfz)-\bfbeta_s(\bfz)\} \vspace{2mm}\nonumber \\
= & \, \frac{ \{\bfbeta_t(\bfz)\}^T \Sigma(\bfz) \bfbeta_t(\bfz) -2 \{\bfbeta_t(\bfz)\}^T \Sigma(\bfz) \bfbeta_t(\bfz)}{\pi_t}\nonumber\\
& + \frac{\{\bfbeta_s(\bfz)\}^T \Sigma(\bfz) \bfbeta_s(\bfz) - 2 \{\bfbeta_s(\bfz)\}^T \Sigma(\bfz) \bfbeta_s(\bfz)  }{\pi_s}   + \{ \bfbeta_t(\bfz)- \bfbeta_s(\bfz)\}^T \Sigma(\bfz)\{\bfbeta_t(\bfz)-\bfbeta_s(\bfz)\} \vspace{2mm}\nonumber \\
=&\,  - \frac{ \{\bfbeta_t(\bfz)\}^T \Sigma(\bfz) \bfbeta_t(\bfz) }{\pi_t} -\frac{ \{\bfbeta_s(\bfz)\}^T \Sigma(\bfz) \bfbeta_s(\bfz) }{\pi_s}+\{ \bfbeta_t(\bfz)- \bfbeta_s(\bfz)\}^T \Sigma(\bfz)\{\bfbeta_t(\bfz)-\bfbeta_s(\bfz)\}    \vspace{2mm}\\
= &\,   -\frac{\{\pi_s \bfbeta_t(\bfz)+\pi_t\bfbeta_s(\bfz)\}^T \Sigma(\bfz)\{\pi_s \bfbeta_t(\bfz)+\pi_t\bfbeta_s(\bfz)\}}{\pi_t\pi_s (\pi_t+\pi_s)}\vspace{2mm}  \\
&-\{ \bfbeta_t(\bfz)- \bfbeta_s(\bfz)\}^T \Sigma(\bfz)\{\bfbeta_t(\bfz)-\bfbeta_s(\bfz)\} \left(\frac{1-\pi_t-\pi_s}{\pi_t+\pi_s}\right)\nonumber
\end{align*}
where the second equality follows  from $\bfbeta_t (\bfz ) = \Sigma(\bfz )^{-1} {\rm cov} (\bX , Y^{(t)} \mid \bfZ = \bfz )$.  
This also proves that $\sigma_A^2 \leq \sigma_U^2$, because 
$\Sigma (\bfz )$ is positive definite for every $\bfz$ and $ \pi_t+\pi_s\leq 1 $. If $\sigma_A^2 = \sigma_U^2$, then we must have $\pi_s \bfbeta_t(\bfz)+\pi_t\bfbeta_s(\bfz) =0$ and $(1-\pi_t-\pi_s) \{\beta_t(z)- \beta_s(z)\}=0  $  for every $\bfz$, which is the same as (\ref{equiv1}).

To show the result for $ \sigma_B^2- \sigma_A^2 $ , note that 
$ \sigma^2_B =\sum_{\bfz \in {\cal Z}} \P(\bfZ = \bfz ) \sigma_B^2(\bfz )$,
where 
\begin{align*}
\sigma^2_B (\bfz )=& 
\frac{\var \{Y^{(t)}- \bX^T \! \bfbeta (\bfz )\mid\bfZ=\bfz\}}{\pi_t} + 
\frac{\var \{Y^{(s)}- \bX^T \! \bfbeta (\bfz )\mid\bfZ=\bfz\}}{\pi_s}\\
=& \frac{\var \{Y^{(t)}- \bX^T \! \bfbeta_t(z)+ \bX^T \! \bfbeta_t(z)- \bX^T \! \bfbeta (\bfz )\mid\bfZ=\bfz\}}{\pi_t} \\
&+  \frac{\var \{Y^{(s)}- \bX^T \! \bfbeta_s(z)+ \bX^T \! \bfbeta_s(z)- \bX^T \! \bfbeta (\bfz )\mid\bfZ=\bfz\}}{\pi_s} \\
=&\frac{\var \{Y^{(t)}- \bX^T \! \bfbeta_t(z)\mid\bfZ=\bfz\}+ \var\{\bX^T \! \bfbeta_t(z)- \bX^T \! \bfbeta (\bfz )\mid\bfZ=\bfz\}}{\pi_t} \\
&+ \frac{\var \{Y^{(s)}- \bX^T \! \bfbeta_s(z) \mid Z=z\}+\var\{ \bX^T \! \bfbeta_s(z)- \bX^T \! \bfbeta (\bfz )\mid\bfZ=\bfz\}}{\pi_s} 
\end{align*}
where the second equality is because 
\begin{align*}
&{\rm cov} \{ Y^{(t)}- \bfbeta_t(z)^T X, \bfbeta_t(z)^T X- \bfbeta(z)^T X\mid Z=z  \}\\
&={\rm cov} \{ Y^{(t)}-\bfbeta_t(z)^T X, \bX \mid Z=z  \} \{\bfbeta_t(z)-\bfbeta (\bfz )\} \\
&= \{ {\rm cov} ( Y^{(t)}, \bX  \mid Z=z  )- \bfbeta_t(z)^T \var( \bX  \mid Z=z)\} \{\bfbeta_t(z)-\bfbeta (\bfz )\} =0
\end{align*}

Then,
\begin{align*}
&\sigma_A^2(z)- \sigma_B^2(z)=  \{\beta_t(\bfz)- \bfbeta_s(\bfz)\}^T  \Sigma(\bfz)\{\bfbeta_t(\bfz)- \bfbeta_s(\bfz)\}\\
&-\frac{ \{ \beta_t(z)- \beta(z)\}^T\Sigma(z)\{ \beta_t(z)- \beta(z)\}}{\pi_t}- \frac{ \{ \beta_s(z)- \beta(z)\}^T\Sigma(z)\{ \beta_s(z)- \beta(z)\}}{\pi_s}
\end{align*}

In order to show that $ \sigma_A^2(z)-\sigma_B^2(z)\leq 0  $ for every $ z $, we prove a stronger statement:  for each given $ z $,  it is true that for any $ \tilde\beta(z) $,
\begin{align}
&  \{\beta_t(\bfz)- \bfbeta_s(\bfz)\}^T  \Sigma(\bfz)\{\bfbeta_t(\bfz)- \bfbeta_s(\bfz)\}\nonumber\\
&-\frac{ \{ \beta_t(z)-\tilde \beta(z)\}^T\Sigma(z)\{ \beta_t(z)- \tilde\beta(z)\}}{\pi_t}- \frac{ \{ \beta_s(z)- \tilde\beta(z)\}^T\Sigma(z)\{ \beta_s(z)-\tilde \beta(z)\}}{\pi_s}\leq 0. \label{eq: state}
\end{align}
As a consequence, setting $ \tilde\beta(z) $ as $ \beta(z) = \sum_t\pi_t\beta_t(z)$, the statement in (\ref{eq: state}) also holds. This proves $ \sigma_A^2(z)- \sigma_B^2(z)\leq 0 $. 

In what follows, we prove the claim in (\ref{eq: state}). For each given $ z $, the gradient of the left hand side of (\ref{eq: state}) is 
\begin{align*}
- 2\left[ \frac{\{\tilde\beta(z)- \beta_t(z) \}^T \Sigma(z) }{\pi_t}+\frac{\{\tilde\beta(z)- \beta_s(z) \}^T \Sigma(z) }{\pi_s} \right],
\end{align*}
which equals zero when $ \tilde\beta(z)= \{ \pi_s\beta_t(z)+ \pi_t\beta_s(z)\}/(\pi_t+\pi_s) $. This is also the unique solution from the positive definiteness of $ \Sigma(z) $. It is also easy to see that the Hessian of the left hand side of (\ref{eq: state}) is negative definite, which means that $ \tilde\beta(z)= \{ \pi_s\beta_t(z)+ \pi_t\beta_s(z)\}/(\pi_t+\pi_s) $ is the global and unique maximizer of the left hand side of (\ref{eq: state}). The statement in (\ref{eq: state}) is true because when evaluated at $ \tilde\beta(z)= \{ \pi_s\beta_t(z)+ \pi_t\beta_s(z)\}/(\pi_t+\pi_s) $, the left hand side of (\ref{eq: state}) equals
\begin{align*}
&  \{\beta_t(\bfz)- \bfbeta_s(\bfz)\}^T  \Sigma(\bfz)\{\bfbeta_t(\bfz)- \bfbeta_s(\bfz)\}\nonumber\\
&- \left\{\beta_t(z)- \frac{\pi_s\beta_t(z) + \pi_t\beta_s(z)}{\pi_t+\pi_s}\right\}^T \Sigma(z)  \left\{\beta_t(z)- \frac{\pi_s\beta_t(z) + \pi_t\beta_s(z)}{\pi_t+\pi_s}\right\}\frac{1}{\pi_t}\\
&- \left\{\beta_s(z)- \frac{\pi_s\beta_t(z) + \pi_t\beta_s(z)}{\pi_t+\pi_s}\right\}^T \Sigma(z)  \left\{\beta_s(z)- \frac{\pi_s\beta_t(z) + \pi_t\beta_s(z)}{\pi_t+\pi_s}\right\}\frac{1}{\pi_s}\\
=& -\{\beta_t(\bfz)- \bfbeta_s(\bfz)\}^T  \Sigma(\bfz)\{\bfbeta_t(\bfz)- \bfbeta_s(\bfz)\} \left(\frac{1-\pi_t-\pi_s}{\pi_t+\pi_s}\right)\leq 0\nonumber
\end{align*}
This completes the proof for $ \sigma_A^2\leq \sigma_B^2 $, where the equality holds if and only if $\{\beta_t(z) - \beta_s(z)\}(1- \pi_t-\pi_s)=0$ and $ \sum_{\ell=1}^k \pi_\ell\beta_\ell(z)= \{ \pi_s\beta_t(z)+ \pi_t\beta_s(z)\}/(\pi_t+\pi_s)  $ for every $ z $, which is the same as (\ref{equiv2}).

\subsection*{Proof of the Consistency of Variance Estimators}

First, we prove the consistency of estimators $ \hat\sigma_U^2, \hat\sigma_V^2, \hat\sigma_A^2$ and  $\hat\sigma_B^2 $. For $\hat{\sigma}_U^2$, following the same arguments as in the proof of Lemma 1, we can show that 
$$S_t^2(\bfz)=\var(Y^{(t)}\mid\bfZ=\bfz)+o_p(1)$$  From $n(\bfz)/n=\P(\bfZ=\bfz)+o_p(1)$,  we conclude that  $\hat{\sigma}_U^2=\sigma_U^2+o_p(1)$ by continuous mapping theorem. For $\hat{\sigma }_V^2$, following the same arguments as in the proof of Lemma 1, we conclude that 
$$\bar{Y}_t(\bfz)=E(Y^{(t)}\mid\bfZ=\bfz)+o_p(1), $$
which together with  the fact that $\hat\theta=\theta+o_p(1)$ implied by Theorem 1, we obtain  that 
\begin{align*}
\hat\sigma_V^2 & = \sum_{z\in {\cal Z}} \P(\bfZ=\bfz) \left\{ E(Y^{(t)}-Y^{(s)}\mid\bfZ=\bfz)\right\}^2-\theta^2+o_p(1) \\
& = E\left[\left\{ E(Y^{(t)}-Y^{(s)}\mid\bfZ)\right\}^2\right]- \{E(Y^{(t)}-Y^{(s)})\}^2+o_p(1)\\
& =\sigma_V^2+o_p(1).
\end{align*}
The proof for the consistency of $ \hat\sigma_A^2 $ and $ \hat\sigma_B^2 $ with $\hat{\bfbeta}_t(\bfz), \hat\beta(z)$ respectively replaced with $\bfbeta_t(\bfz), \beta(z)$ is analogous and is  omitted. From this,  the consistency of  $ \hat\sigma_A^2 $ and $ \hat\sigma_B^2 $  can be  established using Lemma \ref{lemma: coef}.


\section*{Additional Simulation Results}
Tables S1 and S2 report the bias, standard deviation (SD), average estimated SD (SE), and coverage probability (CP) of asymptotic 95\% confidence interval, estimate $\pm 1.96 \,$SE, of $\hat\theta$, $\hat\theta_A$, and $\hat\theta_B$ for cases I-III under stratified permuted block randomization and urn design, respectively, based on 2,000 simulation runs.

\begin{table}[h!]
	\begin{center}
		\caption{Bias, standard deviation (SD), average estimated SD (SE), and coverage probability (CP) of 95\% asymptotic confidence interval under Stratified Permuted Block Randomization for cases I-III (block size = 4 for 1:1 allocation and block size = 6 for 1:2 allocation)}
		\resizebox{0.9\textwidth}{!}{
			\begin{tabular}{cccccccccccccc}
				\hline \\[-1.5ex]
				&      &              &                &  & \multicolumn{4}{c}{treatment   allocation 1:1} &  & \multicolumn{4}{c}{treatment   allocation 1:2} \\[0.5ex]
				\cline{6-9} \cline{11-14} \\[-1.5ex]
				$n$  & case & $Z$          & estimator      &  & bias       & SD        & SE        & CP        &  & bias       & SD        & SE        & CP        \\[0.5ex]
				\hline \\[-1.5ex]
				500 & I    & $X_1$        & $\hat\theta$   &  & 0.0042     & 0.2005    & 0.2000    & 0.9495    &  & 0.0011     & 0.2161    & 0.2124    & 0.9500    \\
				&      &              & $\hat\theta_B$ &  & 0.0016     & 0.0887    & 0.0894    & 0.9515    &  & -0.0001    & 0.0967    & 0.0948    & 0.9480    \\
				&      &              & $\hat\theta_A$ &  & 0.0016     & 0.0887    & 0.0894    & 0.9520    &  & 0.0000     & 0.0969    & 0.0947    & 0.9495    \\
				&      & $X_1$, $D_2$ & $\hat\theta$   &  & -0.0052    & 0.1473    & 0.1467    & 0.9455    &  & 0.0020     & 0.1552    & 0.1556    & 0.9530    \\
				&      &              & $\hat\theta_B$ &  & -0.0004    & 0.0907    & 0.0894    & 0.9430    &  & 0.0014     & 0.0966    & 0.0948    & 0.9445    \\
				&      &              & $\hat\theta_A$ &  & -0.0003    & 0.0908    & 0.0894    & 0.9430    &  & 0.0015     & 0.0964    & 0.0946    & 0.9440    \\
				&      & $X_1$, $D_4$ & $\hat\theta$   &  & 0.0024     & 0.1121    & 0.1152    & 0.9520    &  & 0.0009     & 0.1256    & 0.1220    & 0.9420    \\
				&      &              & $\hat\theta_B$ &  & 0.0009     & 0.0884    & 0.0895    & 0.9430    &  & -0.0026    & 0.0992    & 0.0947    & 0.9360    \\
				&      &              & $\hat\theta_A$ &  & 0.0009     & 0.0885    & 0.0895    & 0.9465    &  & -0.0025    & 0.1002    & 0.0944    & 0.9320    \\
				& II   & $X_1$        & $\hat\theta$   &  & 0.0004     & 0.2168    & 0.2191    & 0.9505    &  & 0.0030     & 0.2329    & 0.2304    & 0.9400    \\
				&      &              & $\hat\theta_B$ &  & 0.0002     & 0.2173    & 0.2184    & 0.9485    &  & 0.0029     & 0.2600    & 0.2543    & 0.9415    \\
				&      &              & $\hat\theta_A$ &  & 0.0005     & 0.2169    & 0.2190    & 0.9510    &  & 0.0034     & 0.2247    & 0.2214    & 0.9405    \\
				&      & $X_1$, $D_2$ & $\hat\theta$   &  & 0.0039     & 0.2199    & 0.2190    & 0.9500    &  & 0.0006     & 0.2275    & 0.2255    & 0.9405    \\
				&      &              & $\hat\theta_B$ &  & 0.0037     & 0.2207    & 0.2177    & 0.9500    &  & 0.0004     & 0.2393    & 0.2350    & 0.9410    \\
				&      &              & $\hat\theta_A$ &  & 0.0038     & 0.2203    & 0.2189    & 0.9490    &  & 0.0007     & 0.2241    & 0.2216    & 0.9440    \\
				&      & $X_1$, $D_4$ & $\hat\theta$   &  & 0.0020     & 0.2131    & 0.2193    & 0.9570    &  & -0.0046    & 0.2246    & 0.2230    & 0.9465    \\
				&      &              & $\hat\theta_B$ &  & 0.0021     & 0.2140    & 0.2176    & 0.9515    &  & -0.0058    & 0.2314    & 0.2253    & 0.9415    \\
				&      &              & $\hat\theta_A$ &  & 0.0017     & 0.2130    & 0.2190    & 0.9560    &  & -0.0034    & 0.2239    & 0.2211    & 0.9405    \\
				& III  & $X_1$        & $\hat\theta$   &  & 0.0046     & 0.1737    & 0.1732    & 0.9500    &  & 0.0028     & 0.1674    & 0.1666    & 0.9465    \\
				&      &              & $\hat\theta_B$ &  & 0.0048     & 0.1485    & 0.1477    & 0.9490    &  & 0.0037     & 0.1688    & 0.1656    & 0.9395    \\
				&      &              & $\hat\theta_A$ &  & 0.0049     & 0.1484    & 0.1479    & 0.9500    &  & 0.0057     & 0.1550    & 0.1532    & 0.9455    \\
				&      & $X_1$, $D_2$ & $\hat\theta$   &  & 0.0028     & 0.1615    & 0.1593    & 0.9450    &  & 0.0016     & 0.1620    & 0.1597    & 0.9455    \\
				&      &              & $\hat\theta_B$ &  & 0.0083     & 0.1501    & 0.1471    & 0.9425    &  & 0.0044     & 0.1626    & 0.1579    & 0.9380    \\
				&      &              & $\hat\theta_A$ &  & 0.0064     & 0.1499    & 0.1475    & 0.9440    &  & 0.0039     & 0.1562    & 0.1526    & 0.9380    \\
				&      & $X_1$, $D_4$ & $\hat\theta$   &  & 0.0033     & 0.1500    & 0.1527    & 0.9540    &  & -0.0012    & 0.1565    & 0.1560    & 0.9500    \\
				&      &              & $\hat\theta_B$ &  & 0.0062     & 0.1450    & 0.1469    & 0.9515    &  & 0.0001     & 0.1586    & 0.1540    & 0.9410    \\
				&      &              & $\hat\theta_A$ &  & 0.0031     & 0.1448    & 0.1475    & 0.9520    &  & -0.0006    & 0.1558    & 0.1521    & 0.9400    \\
				&      &              &                &  &            &           &           &           &  &            &           &           &           \\[-1.5ex]
				100 & I    & $X_1$        & $\hat\theta$   &  & -0.0019    & 0.4505    & 0.4484    & 0.9460    &  & -0.0087    & 0.4690    & 0.4741    & 0.9500    \\
				&      &              & $\hat\theta_B$ &  & -0.0004    & 0.2004    & 0.1977    & 0.9350    &  & 0.0013     & 0.2132    & 0.2094    & 0.9420    \\
				&      &              & $\hat\theta_A$ &  & -0.0005    & 0.2004    & 0.1978    & 0.9365    &  & 0.0013     & 0.2147    & 0.2084    & 0.9395    \\
				&      & $X_1$, $D_2$ & $\hat\theta$   &  & -0.0072    & 0.3305    & 0.3314    & 0.9500    &  & -0.0043    & 0.3414    & 0.3498    & 0.9500    \\
				&      &              & $\hat\theta_B$ &  & -0.0003    & 0.2036    & 0.1983    & 0.9415    &  & 0.0018     & 0.2184    & 0.2094    & 0.9340    \\
				&      &              & $\hat\theta_A$ &  & -0.0017    & 0.2059    & 0.1997    & 0.9380    &  & 0.0054     & 0.2387    & 0.2113    & 0.9165    \\
				& II   & $X_1$        & $\hat\theta$   &  & 0.0151     & 0.4930    & 0.4901    & 0.9500    &  & 0.0167     & 0.5103    & 0.5140    & 0.9445    \\
				&      &              & $\hat\theta_B$ &  & 0.0136     & 0.4992    & 0.4827    & 0.9430    &  & 0.0230     & 0.5728    & 0.5607    & 0.9405    \\
				&      &              & $\hat\theta_A$ &  & 0.0149     & 0.4945    & 0.4892    & 0.9470    &  & 0.0152     & 0.4938    & 0.4929    & 0.9455    \\
				&      & $X_1$, $D_2$ & $\hat\theta$   &  & 0.0112     & 0.4953    & 0.4916    & 0.9475    &  & 0.0209     & 0.5070    & 0.5035    & 0.9435    \\
				&      &              & $\hat\theta_B$ &  & 0.0085     & 0.5059    & 0.4781    & 0.9320    &  & 0.0223     & 0.5428    & 0.5130    & 0.9310    \\
				&      &              & $\hat\theta_A$ &  & 0.0102     & 0.4976    & 0.4907    & 0.9415    &  & 0.0226     & 0.5105    & 0.4951    & 0.9410    \\
				& III  & $X_1$        & $\hat\theta$   &  & 0.0033     & 0.3877    & 0.3869    & 0.9485    &  & 0.0016     & 0.3668    & 0.3714    & 0.9460    \\
				&      &              & $\hat\theta_B$ &  & 0.0110     & 0.3325    & 0.3262    & 0.9465    &  & 0.0175     & 0.3692    & 0.3651    & 0.9365    \\
				&      &              & $\hat\theta_A$ &  & 0.0116     & 0.3306    & 0.3287    & 0.9490    &  & 0.0212     & 0.3413    & 0.3390    & 0.9380    \\
				&      & $X_1$, $D_2$ & $\hat\theta$   &  & -0.0006    & 0.3642    & 0.3578    & 0.9430    &  & 0.0071     & 0.3547    & 0.3575    & 0.9485    \\
				&      &              & $\hat\theta_B$ &  & 0.0163     & 0.3404    & 0.3241    & 0.9285    &  & 0.0240     & 0.3653    & 0.3456    & 0.9375    \\
				&      &              & $\hat\theta_A$ &  & 0.0079     & 0.3394    & 0.3296    & 0.9365    &  & 0.0256     & 0.3701    & 0.3400    & 0.9330    \\[0.5ex] 
				\hline
		\end{tabular}}
	\end{center}
\end{table}

\begin{table}[h!]
	\begin{center}
		\caption{Bias, standard deviation (SD), average estimated SD (SE), and coverage probability (CP) of 95\% asymptotic confidence interval under Stratified Urn Design for cases I-III}
		\resizebox{0.9\textwidth}{!}{
			\begin{tabular}{cccccccccccccc}
				\hline \\[-1.5ex]
				&      &              &                &  & \multicolumn{4}{c}{treatment allocation   1:1} &  & \multicolumn{4}{c}{treatment allocation   1:2} \\[0.5ex]
				\cline{6-9} \cline{11-14} \\[-1.5ex]
				$n$   & case & $Z$          & estimator      &  & bias       & SD        & SE        & CP        &  & bias       & SD        & SE        & CP        \\ [0.5ex]
				\hline \\[-1.5ex]
				500 & I    & $X_1$        & $\hat\theta$   &  & -0.0095    & 0.1984    & 0.2000    & 0.9465    &  & 0.0012     & 0.2119    & 0.2121    & 0.9470    \\
				&      &              & $\hat\theta_B$ &  & -0.0052    & 0.0866    & 0.0893    & 0.9585    &  & -0.0012    & 0.0957    & 0.0948    & 0.9455    \\
				&      &              & $\hat\theta_A$ &  & -0.0052    & 0.0867    & 0.0893    & 0.9580    &  & -0.0013    & 0.0957    & 0.0947    & 0.9445    \\
				&      & $X_1$, $D_2$ & $\hat\theta$   &  & 0.0001     & 0.1486    & 0.1466    & 0.9430    &  & 0.0049     & 0.1560    & 0.1554    & 0.9510    \\
				&      &              & $\hat\theta_B$ &  & -0.0004    & 0.0902    & 0.0894    & 0.9405    &  & 0.0020     & 0.0948    & 0.0947    & 0.9440    \\
				&      &              & $\hat\theta_A$ &  & -0.0004    & 0.0902    & 0.0894    & 0.9420    &  & 0.0022     & 0.0953    & 0.0945    & 0.9435    \\
				&      & $X_1$, $D_4$ & $\hat\theta$   &  & -0.0001    & 0.1160    & 0.1150    & 0.9470    &  & 0.0008     & 0.1216    & 0.1220    & 0.9520    \\
				&      &              & $\hat\theta_B$ &  & -0.0019    & 0.0927    & 0.0893    & 0.9390    &  & -0.0013    & 0.0966    & 0.0947    & 0.9445    \\
				&      &              & $\hat\theta_A$ &  & -0.0018    & 0.0930    & 0.0894    & 0.9390    &  & -0.0017    & 0.0993    & 0.0947    & 0.9405    \\
				& II   & $X_1$        & $\hat\theta$   &  & 0.0037     & 0.2196    & 0.2192    & 0.9510    &  & -0.0040    & 0.2331    & 0.2305    & 0.9480    \\
				&      &              & $\hat\theta_B$ &  & 0.0033     & 0.2203    & 0.2186    & 0.9495    &  & -0.0046    & 0.2596    & 0.2544    & 0.9465    \\
				&      &              & $\hat\theta_A$ &  & 0.0034     & 0.2196    & 0.2191    & 0.9495    &  & -0.0034    & 0.2230    & 0.2214    & 0.9475    \\
				&      & $X_1$, $D_2$ & $\hat\theta$   &  & 0.0083     & 0.2200    & 0.2191    & 0.9530    &  & -0.0012    & 0.2261    & 0.2253    & 0.9450    \\
				&      &              & $\hat\theta_B$ &  & 0.0085     & 0.2216    & 0.2180    & 0.9465    &  & -0.0015    & 0.2382    & 0.2346    & 0.9450    \\
				&      &              & $\hat\theta_A$ &  & 0.0082     & 0.2205    & 0.2190    & 0.9480    &  & 0.0001     & 0.2227    & 0.2213    & 0.9455    \\
				&      & $X_1$, $D_4$ & $\hat\theta$   &  & 0.0071     & 0.2214    & 0.2193    & 0.9475    &  & 0.0068     & 0.2246    & 0.2230    & 0.9460    \\
				&      &              & $\hat\theta_B$ &  & 0.0080     & 0.2226    & 0.2176    & 0.9450    &  & 0.0068     & 0.2306    & 0.2254    & 0.9415    \\
				&      &              & $\hat\theta_A$ &  & 0.0068     & 0.2223    & 0.2190    & 0.9465    &  & 0.0066     & 0.2243    & 0.2213    & 0.9440    \\
				& III  & $X_1$        & $\hat\theta$   &  & -0.0037    & 0.1724    & 0.1731    & 0.9470    &  & -0.0015    & 0.1659    & 0.1666    & 0.9530    \\
				&      &              & $\hat\theta_B$ &  & -0.0003    & 0.1480    & 0.1477    & 0.9495    &  & -0.0014    & 0.1694    & 0.1657    & 0.9385    \\
				&      &              & $\hat\theta_A$ &  & -0.0002    & 0.1478    & 0.1479    & 0.9525    &  & 0.0008     & 0.1548    & 0.1533    & 0.9475    \\
				&      & $X_1$, $D_2$ & $\hat\theta$   &  & 0.0029     & 0.1594    & 0.1593    & 0.9510    &  & 0.0013     & 0.1596    & 0.1594    & 0.9465    \\
				&      &              & $\hat\theta_B$ &  & 0.0052     & 0.1492    & 0.1470    & 0.9495    &  & 0.0029     & 0.1588    & 0.1576    & 0.9420    \\
				&      &              & $\hat\theta_A$ &  & 0.0031     & 0.1487    & 0.1474    & 0.9510    &  & 0.0028     & 0.1530    & 0.1523    & 0.9420    \\
				&      & $X_1$, $D_4$ & $\hat\theta$   &  & 0.0012     & 0.1534    & 0.1525    & 0.9455    &  & 0.0032     & 0.1569    & 0.1558    & 0.9515    \\
				&      &              & $\hat\theta_B$ &  & 0.0049     & 0.1508    & 0.1468    & 0.9445    &  & 0.0060     & 0.1572    & 0.1539    & 0.9460    \\
				&      &              & $\hat\theta_A$ &  & 0.0014     & 0.1510    & 0.1473    & 0.9455    &  & 0.0039     & 0.1554    & 0.1521    & 0.9520    \\
				&      &              &                &  &            &           &           &           &  &            &           &           &           \\[-1.5ex]
				100 & I    & $X_1$        & $\hat\theta$   &  & 0.0216     & 0.4468    & 0.4487    & 0.9470    &  & -0.0132    & 0.4801    & 0.4757    & 0.9460    \\
				&      &              & $\hat\theta_B$ &  & 0.0012     & 0.2033    & 0.1981    & 0.9365    &  & -0.0070    & 0.2115    & 0.2103    & 0.9495    \\
				&      &              & $\hat\theta_A$ &  & 0.0014     & 0.2034    & 0.1982    & 0.9390    &  & -0.0073    & 0.2145    & 0.2093    & 0.9420    \\
				&      & $X_1$, $D_2$ & $\hat\theta$   &  & -0.0102    & 0.3299    & 0.3306    & 0.9465    &  & -0.0070    & 0.3562    & 0.3497    & 0.9425    \\
				&      &              & $\hat\theta_B$ &  & -0.0023    & 0.2073    & 0.1980    & 0.9380    &  & -0.0051    & 0.2200    & 0.2096    & 0.9295    \\
				&      &              & $\hat\theta_A$ &  & -0.0020    & 0.2137    & 0.1999    & 0.9340    &  & -0.0352    & 0.7085    & 0.2280    & 0.9015    \\
				& II   & $X_1$        & $\hat\theta$   &  & -0.0076    & 0.4882    & 0.4903    & 0.9460    &  & -0.0131    & 0.5116    & 0.5141    & 0.9430    \\
				&      &              & $\hat\theta_B$ &  & -0.0073    & 0.4993    & 0.4841    & 0.9380    &  & -0.0143    & 0.5838    & 0.5617    & 0.9400    \\
				&      &              & $\hat\theta_A$ &  & -0.0072    & 0.4890    & 0.4895    & 0.9475    &  & -0.0159    & 0.4929    & 0.4938    & 0.9440    \\
				&      & $X_1$, $D_2$ & $\hat\theta$   &  & -0.0103    & 0.4869    & 0.4903    & 0.9485    &  & -0.0113    & 0.5099    & 0.5029    & 0.9420    \\
				&      &              & $\hat\theta_B$ &  & -0.0121    & 0.4990    & 0.4776    & 0.9375    &  & -0.0090    & 0.5470    & 0.5125    & 0.9270    \\
				&      &              & $\hat\theta_A$ &  & -0.0108    & 0.4908    & 0.4895    & 0.9450    &  & -0.0435    & 0.8393    & 0.5090    & 0.9325    \\
				& III  & $X_1$        & $\hat\theta$   &  & 0.0057     & 0.3864    & 0.3869    & 0.9445    &  & -0.0174    & 0.3722    & 0.3726    & 0.9475    \\
				&      &              & $\hat\theta_B$ &  & 0.0038     & 0.3348    & 0.3270    & 0.9380    &  & -0.0071    & 0.3754    & 0.3658    & 0.9425    \\
				&      &              & $\hat\theta_A$ &  & 0.0040     & 0.3314    & 0.3291    & 0.9420    &  & 0.0006     & 0.3435    & 0.3396    & 0.9420    \\
				&      & $X_1$, $D_2$ & $\hat\theta$   &  & -0.0104    & 0.3566    & 0.3569    & 0.9505    &  & -0.0078    & 0.3625    & 0.3569    & 0.9445    \\
				&      &              & $\hat\theta_B$ &  & 0.0067     & 0.3369    & 0.3234    & 0.9310    &  & 0.0082     & 0.3652    & 0.3451    & 0.9310    \\
				&      &              & $\hat\theta_A$ &  & -0.0012    & 0.3369    & 0.3289    & 0.9400    &  & -0.0388    & 1.0589    & 0.3615    & 0.9225    \\[0.5ex] 
				\hline
		\end{tabular}}
	\end{center}
\end{table}

\end{document}